%% file: main.tex
\documentclass[runningheads]{llncs}
\input{header}

\usepackage{wuhao}

\makeatletter
\def\thanks#1{\protected@xdef\@thanks{\@thanks
        \protect\footnotetext{#1}}}
\makeatother

\begin{document}

\title{Quantifier Elimination Meets Treewidth}

\thanks{The first two authors marked with $\star$ contributed equally to this work and should be considered co-first authors. 
}

\authorrunning{H.~Wu et al.}

\author{Hao Wu$^\star$\inst{1,2}\orcidlink{0000-0001-9368-4744} \and
Jiyu Zhu$^\star$\inst{1,2}\orcidlink{0009-0001-1885-0674}  \and
Amir Kafshdar Goharshady\inst{3}\orcidlink{0000-0003-1702-6584} \and 
\\
Jie An\inst{4,2}\orcidlink{0000-0001-9260-9697} \and 
Bican Xia\inst{5}\orcidlink{0000-0002-2570-2338} \and
Naijun Zhan \inst{6,7}\inst{(}\Envelope\inst{)}\orcidlink{0000-0003-3298-3817}
}

\institute{
KLSS, Institute of Software Chinese Academy of Sciences, China
\and
University of Chinese Academy of Sciences, China
\and
University of Oxford, United Kingdom
\and 
National Key Laboratory of Space Integrated Information System, \\
Institute of Software Chinese Academy of Sciences, China
\and
School of Mathematical Sciences, Peking University, China 
\and
School of Computer Science, Peking University, China
\and
Zhongguancun Laboratory, China \\
\email{\{wuhao, zhujy\}@ios.ac.cn} \quad \email{amir.goharshady@cs.ox.ac.uk} \\
\email{anjie@iscas.ac.cn}
\quad \email{xbc@math.pku.edu.cn} 
\quad \email{njzhan@pku.edu.cn}
}
\maketitle

\setcounter{footnote}{0}

\input{sections/0-abstract}
\input{sections/1-introduction}
\input{sections/2-preliminaries}
\input{sections/3-algorithm}

\input{sections/4-analysis}
\input{sections/5-experiments}

\input{sections/6-summary}

\bibliographystyle{splncs04}
\bibliography{bibfiles/hao}

\appendix
\input{sections/appendix}

\end{document}

%% file: header.tex
\usepackage{amsfonts}

\usepackage{latexsym,amsmath,amsthm,epsfig,amssymb,graphics,amscd,amsfonts}
\allowdisplaybreaks
\usepackage{float}
\usepackage{comment}
\usepackage{verbatim}
\usepackage{color}
\usepackage[usenames,dvipsnames,table]{xcolor}
\usepackage{indentfirst}
\usepackage{stmaryrd}
\usepackage{mathrsfs}
\usepackage{algorithm, algorithmic}
\usepackage{amsfonts}
\usepackage{makeidx}
\usepackage{mathtools}
\usepackage{makecell}
\usepackage{multirow}
\usepackage{siunitx}
\usepackage{caption}
\usepackage{subcaption}
\usepackage{wrapfig}

\usepackage[hyperfootnotes=true,colorlinks,linkcolor=RedViolet,anchorcolor=RedViolet,citecolor=blue,urlcolor=black]{hyperref}
\usepackage[capitalize]{cleveref}
\crefname{lemma}{Lem.}{Lem.}
\creflabelformat{lemma}{#2\textup{#1}#3}
\crefname{example}{Exmp.}{Exmp.}
\creflabelformat{example}{#2\textup{#1}#3}
\crefname{section}{Sect.}{Sect.}
\creflabelformat{section}{#2\textup{#1}#3}
\crefname{appendix}{Appx.}{Appx.}
\creflabelformat{appendix}{#2\textup{#1}#3}
\crefname{definition}{Def.}{Def.}
\creflabelformat{definition}{#2\textup{#1}#3}
\crefname{theorem}{Thm.}{Thm.}
\creflabelformat{theorem}{#2\textup{#1}#3}
\crefname{proposition}{Prop.}{Prop.}
\creflabelformat{proposition}{#2\textup{#1}#3}
\crefname{corollary}{Cor.}{Cor.}
\creflabelformat{corollary}{#2\textup{#1}#3}
\crefname{problem}{Problem}{Problem}
\creflabelformat{problem}{#2\textup{#1}#3}
\crefname{algorithm}{Alg.}{Alg.}
\creflabelformat{algorithm}{#2\textup{#1}#3}
\crefname{equation}{Eq.}{Eq.}
\crefname{enumi}{}{}

\usepackage[misc]{ifsym}
\usepackage{orcidlink}
\usepackage{bbding}

\RequirePackage{todonotes,xspace}
\newcommand{\oomit}[1]{ }

\newcommand{\myparagraph}[1]{\medskip\noindent{\bf #1}}

%% file: sections/0-abstract.tex
\begin{abstract}


In this paper, we address the complexity barrier inherent in Fourier-Motzkin elimination (FME) and cylindrical algebraic decomposition (CAD) when eliminating a block of (existential) quantifiers. 
To mitigate this, we propose exploiting structural sparsity in the variable dependency graph of quantified formulas. 
Utilizing tools from parameterized algorithms, we investigate the role of \emph{treewidth}, a parameter that measures the graph's tree-likeness, in the process of quantifier elimination.
A novel dynamic programming framework, structured over a tree decomposition of the dependency graph, is developed for applying FME and CAD, and is also extensible to general quantifier elimination procedures.
Crucially, we prove that when the treewidth is a constant, the framework achieves a significant exponential complexity improvement for both FME and CAD, reducing the worst-case complexity bound from doubly exponential to single exponential.
Preliminary experiments on sparse linear real arithmetic (LRA) and nonlinear real arithmetic (NRA) benchmarks confirm that our algorithm outperforms the existing popular heuristic-based approaches on instances exhibiting low treewidth.


\keywords{Quantifier Elimination \and Fourier-Motzkin Elimination \and Cylindrical Algebraic Decomposition \and Parameterized Algorithms \and Treewidth}
\end{abstract}

%% file: sections/1-introduction.tex
\section{Introduction}

Quantifier elimination (QE) is a fundamental technique in mathematical logic that transforms a first-order formula containing existential ($\exists$) and universal ($\forall$) quantifiers into an equivalent quantifier-free formula. 
A theory that admits a quantifier elimination procedure is highly desirable, as its decidability immediately reduces to the decidability of its quantifier-free fragment. 
This theoretical significance strongly motivates researchers in mathematics and computer science to investigate various theories that admit quantifier elimination procedures.
This paper primarily focuses on quantifier elimination within two central theories over the real numbers: \emph{linear real arithmetic} (LRA), where atomic formulas are defined by strictly linear constraints, and \emph{nonlinear real arithmetic} (NRA), which permits general polynomial constraints. 

For LRA, the \emph{Fourier-Motzkin elimination} (FME) algorithm is the first quantifier elimination procedure, originally proposed by Fourier~\cite{fourier1826-fme} and rediscovered by Motzkin~\cite{motzkin1936-fme}.
This algorithm can be understood as reducing a system of linear inequalities by removing variables one by one.
Its importance primarily stems from its geometric interpretation, which corresponds to the projection of a polyhedron (described by linear inequalities) onto lower-dimensional subspaces~\cite[Sect.~12.2]{schrijver1998book-lp}.
However, a significant challenge of applying the FME algorithm is its high worst-case complexity, which is doubly exponential in the number of quantified variables.
This high complexity results from the rapid, often redundant, proliferation of new constraints during the elimination steps~\cite{imbert90,imbert93ppcp,khachiyan09,jing20casc}. 
For eliminating a block of existential quantifiers, recent works~\cite{promies24fm,promies25lmcs} provide a divide-and-conquer approach to improve the complexity to single exponential. 

For NRA, the first quantifier elimination procedure was developed around the 1930s by Tarski in his seminal work~\cite{tarski1930}.
While Tarski's procedure is not elementary recursive, the first elementary recursive quantifier elimination procedure was developed by Collins in 1975~\cite{collins75-cad}, known as \emph{cylindrical algebraic decomposition} (CAD).
The CAD algorithm marks a milestone in real algebraic geometry and a comprehensive survey can be found in~\cite{caviness1998book-qe}.
Ever since its origin, the CAD algorithm, as well as its various variants~\cite{mccallum88jsc,hong90issac,collins91,brown01jsc,chen09issac,hong12jsc,brown15issac,han16jsc,strzebonski16jsc,bradford16jsc,han17jsc}, has remained the most authoritative procedure for NRA. 
It has found broad application in critical domains, including formal verification~\cite{kapur06jssc,rodriguez04issac,liu11emsoft,gan18tac,AnZLZY18}, control synthesis~\cite{dorato97jsc,jirstrand97jsc,hong97jsc}, and hybrid system analysis~\cite{anai01hscc,sturm11issac}.
Similar to the FME algorithm, the CAD algorithm also suffers from a worst-case doubly exponential complexity, but in the number of \emph{all} occurring variables~\cite{davenport88jsc}. 
Besides the complexity of internal algebraic operations, it is well known that the order of variables to be eliminated, called \emph{variable elimination ordering}, has a huge impact on the practical performance when eliminating a block of existential quantifiers.
Therefore, many heuristics based on sophisticated structural analysis
\cite{brown04issac,dolzmann04issac,bradford13cicm,li23jsc} and machine learning techniques have been applied to this task~\cite{huang14cicm,england19cicm,chen20icms,jia23nips}.

While existing works mostly focus on improving FME and CAD for general inputs, a natural but unexplored (up to our knowledge) problem is: \textbf{Can we design strategies for certain classes of inputs with better worst-case complexity upper bounds?}
Specifically, in this paper, we consider posing restrictions on the structural complexity of the dependency relationship of variables.
This relationship is typically represented by a graph, precisely the \emph{primal graph}, whose nodes correspond to variables of the input formula, and two variables are linked if and only if they both occur in an atomic formula.
To study this problem, we utilize tools coming from \emph{parameterized algorithms}.

The study of parameterized algorithms offers a fine-grained analytical framework for tackling computationally hard problems by introducing a secondary measurement, $k$, known as the parameter~\cite{downey12book-parameterized,cygan15book-parameterized}. 
For example, a problem is classified as \emph{fixed-parameter tractable} (FPT) if it admits an algorithm with a time complexity of $O(f(k) \cdot n^c)$, where $n$ is the input size, $f$ is a computable function depending solely on $k$, and $c$ is a constant independent of $n$.
In general, the goal of parameterized complexity is to conduct a finer-grained analysis of complexity by isolating the combinatorial explosion to a specific parameter.

When the input is a graph, \emph{treewidth} is one of the most important parameters to measure its inherent complexity.
Intuitively, treewidth quantifies how ``tree-like'' a graph is: a smaller treewidth value indicates a greater sparsity in the graph structure.
The main advantage of treewidth is that a bounded treewidth allows a vast collection of classical NP-hard problems to become fixed-parameter tractable~\cite{robertson86-minor2,bodlaender97mfcs,cygan15book-parameterized}. 
This tractability is fundamentally realized by employing \emph{dynamic programming} (DP) techniques executed over a tree structure related to the graph, called a \emph{tree decomposition}~\cite{bodlaender88icalp-dp}.
For constraint satisfiability problems, treewidth measures the structural complexity of variable dependencies and has been extensively studied~\cite{szeider03dam,marx10tc,fichte20cp,dong21stoc,gu22open-sdp,mallach25acta-ilp}.
However, the role of treewidth in the quantifier elimination of first-order logic theories is largely untouched, possibly due to the inherent high complexity of these procedures.

\myparagraph{Contributions.}
In this paper, we develop a novel framework to exploit the treewidth sparsity pattern in the process of applying quantifier elimination procedures to eliminate a block of (existential) quantifiers.
Here, the treewidth sparsity pattern means that the primal graph of the input formula has a small treewidth. 
The primary contributions include:
\begin{itemize}
    \item[$\bullet$] We propose a dynamic programming algorithm for applying FME and CAD to eliminate a block of quantifiers. It is executed over a tree decomposition of the formula's primal graph, where the treewidth of the graph determines the structure of this decomposition. We prove its correctness and demonstrate its extensibility to general quantifier elimination procedures. 
    \item[$\bullet$] We prove that, when the treewidth is a constant, the worst-time complexity of FME and CAD can be \emph{exponentially improved} from doubly exponential to single exponential in the number of \emph{eliminated} variables and \emph{all occurring} variables, respectively. The core idea is the utilization of a special tree decomposition structure, called \emph{balanced tree decompositions}~\cite{chatterjee18toplas}.
    \item[$\bullet$] We conduct experiments on randomly generated LRA and NRA benchmarks exhibiting sparse patterns. Experimental results indicate that our algorithm outperforms other popular heuristics on problem instances with low treewidth. 
\end{itemize}

Our starting point is a parameterized algorithmic perspective, a methodology that is orthogonal to most existing techniques. 
The most closely related work is \cite{li23jsc}, which exploits the chordal structure of the primal graph to guide the variable elimination ordering in CAD. 
Since chordal graphs and tree decompositions are two sides of one coin~\cite[Sec.~12.3]{diestel00book-graph}, our algorithm in~\cref{sec:algo} can be viewed as a dual version of theirs.
Nevertheless, utilizing treewidth offers significant algorithmic advantages. Specifically, it enables a dynamic programming framework that is more comprehensible and, by drawing upon concepts from parameterized complexity theory, a more fine-grained analysis of the complexity issues involved.

\myparagraph{Outline.}
The rest of this paper is organized as follows: 
\cref{sec:pre} introduces necessary concepts. 
\cref{sec:algo} presents the dynamic programming algorithm for FME and presents how to extend it to CAD and general quantifier elimination procedures.
\cref{sec:complexity} analyzes the complexity of our framework for FME and CAD under the assumption that the treewidth is a constant.
\cref{sec:exp} reports the experimental results and \cref{sec:summary} finally concludes the paper. 

%% file: sections/2-preliminaries.tex
\section{Preliminaries}
\label{sec:pre}

Let $\Real$ and $\Nat$ denote the set of real numbers and natural numbers, respectively.
A vector of $n$ real variables is denoted by $(x_1,\dots,x_n)\in \Real^n$, also written as $\seq{x}$ for short.
We denote by $\langle x_{\sigma(1)},\dots,x_{\sigma(n)}\rangle$ a linear ordering among these variables, where $\sigma$ is from the permutation group of size~$n$. 
For example, $\langle x_1,\dots,x_n \rangle$ denotes the natural order from $x_1$ to $x_n$. 
We assume readers are familiar with first-order logic and recommend the book~\cite{bradley07book-computation} for reference. For a set $C$ of atomic formulas, we denote by $\bigwedge C$ the conjunction of all atomic formulas in $C$.

\subsection{LRA, NRA, and Quantifier Elimination}
\label{ssec:logic}

\myparagraph{Quantifier Elimination (QE).} 
We say a first-order theory $\mathcal{T}$ admits \emph{quantifier elimination} if there exists an algorithm, called a quantifier elimination procedure, that transforms a given quantified $\mathcal{T}$-formula into an equivalent formula without quantifiers. 
Formally, w.l.o.g., we consider the input formula~$\Phi$ to be a conjunction of atomic formulas that is existentially quantified:
\begin{equation}\label{eq:input}
    \Phi \defeq \exists x_m,\dots,\exists x_1.~\bigwedge_i \varphi_i(x_1,\dots,x_n),
\end{equation}
where $m,n\in \Nat$, $m\le n$, and each $\varphi_i$ is an atomic $\mathcal{T}$-formula. 
The output of the quantifier elimination procedure is a quantifier-free formula in variables $x_{m+1},\dots,x_n$ that is equivalent to $\Phi$. 
For formulas with quantifier alternations (e.g., $\forall x_2\exists x_1$), 
we eliminate quantifiers from inside out, eliminating a block of quantifiers of the same type at each step.
In what follows, we primarily focus on the FME algorithm for LRA and on the CAD algorithm for NRA.


\myparagraph{Linear Real Arithmetic (LRA).}
LRA refers to the first-order theory with signature $\{0,1,+,<\}$ and domain $\Real$.
Given variables $(x_1,\dots,x_n)\in \Real^n$ and real constants $(a_1,\dots,a_n,b)\in \Real^{n+1}$, an LRA atomic formula is of the form 
$a_1x_1+\dots + a_n x_n - b\bowtie 0,$
where the relation symbol $\bowtie\ \in\{<,>,=,\le, \ge\}$.
The semantics of LRA is interpreted in the standard way. 

\myparagraph{Fourier-Motzkin Elimination (FME).} 
The FME algorithm takes a pair $(C, x_r)$ as input, where $C$ is a set of LRA atomic formulas and $x_r$ is a variable occurring in these constraints, and proceeds one of the following two steps:
\begin{itemize}
    \item[(I)] If $x_r$ appears in an equality constraint of the form $\sum_{j=1}^n a_{i,j}\cdot x_j - b_i = 0$ with $a_{i,r}\neq 0$, the procedure removes this constraint from $C$ and outputs $C$ by replacing every occurrence of $x_r$ with  
        $\frac{b_i}{a_{i,r}} - \sum_{j=1}^{r-1}\frac{a_{i,j}}{a_{i,r}} \cdot x_j- \sum_{j=r+1}^{n}\frac{a_{i,j}}{a_{i,r}} \cdot x_j$. 
    \item[(II)] If $x_r$ only appears in inequality constraints of the form $\sum_{j=1}^n a_{i,j}\cdot x_j - b_i \bowtie 0$, where $\bowtie$ is not $=$ and $a_{i,r}\neq 0$, for each of such constraints we derive a term $\frac{b_i}{a_{i,r}} - \sum_{j=1}^{r-1}\frac{a_{i,j}}{a_{i,r}} \cdot x_j- \sum_{j=r+1}^{n}\frac{a_{i,j}}{a_{i,r}} \cdot x_j$ that is called a \emph{bound} of $x_r$.
    Moreover, the bound is called an upper bound if $a_{i,r}>0$ and a lower bound if $a_{i,r}<0$. In either case, it is called strict if $\bowtie\ \in \{<,>\}$.
    Let $U_r$ and $L_r$ denote the set of upper bounds and lower bounds of $x_r$, respectively.
    Then, the procedure outputs the set of constraints 
    $\set{l \lhd u\given l\in L_r, u\in U_r}$, where $\lhd$ is $<$ if both $u$ and $l$ are strict and is $\le$ otherwise.
    
\end{itemize}
We denote the output by $\textsf{FME}(C, x_r)$, which is a set of atomic formulas without $x_r$.
If $C$ is taken to be the set of atomic formulas of an LRA formula $\Phi$ of the form \cref{eq:input}, then $\Phi$ is equivalent to the following quantifier-free formula 
    $\bigwedge \textsf{FME}(C, \langle x_1,\dots, x_m\rangle)$,
where $\textsf{FME}(C, \langle x_1,\dots, x_m\rangle)$ denotes the output of recursively applying $\textsf{FME}$ procedure on $C$ w.r.t. the ordering from $x_1$ to $x_m$.
For a more detailed description of the algorithm, please refer to \cite [Sec.~5.4]{daniel08book}.

\myparagraph{Nonlinear Real Arithmetic (NRA).}
NRA refers to the first-order theory with signature $\{0,1,+,\cdot,<\}$ and domain $\Real$.
A monomial is a product of powers of variables with nonnegative integer exponents, denoted by $\seq{x}^\seq{\alpha}=x_1^{\alpha_1}\dots x_n^{\alpha_n}$ for some $\alpha\in \Nat^n$. 
A polynomial is a finite summation of monomials $\sum_{\alpha} c_{\alpha} \seq{x}^\seq{\alpha}$, where $c_{\alpha}\in \Real$ are called coefficients. 
The set of polynomials in variables $\seq{x}$ with real coefficients is denoted by $\Real[\seq{x}]$.
An NRA atomic formula is of the form $f(\seq{x}) \bowtie 0$, where $f(\seq{x})\in \Real[\seq{x}]$ and $\bowtie\ \in \{<,>,=,\le, \ge\}$.
It is clear that an LRA constraint is a special case of an NRA constraint where the polynomial $f$ is linear, i.e., of degree $1$.
The semantics of NRA are also interpreted in the standard way. 

\myparagraph{Cylindrical Algebraic Decomposition (CAD).}
Given an NRA formula $\Phi$ of the form \cref{eq:input} and letting $P\subset \Real[x_1,\dots,x_n]$ denote the set of polynomials occurring in~$\Phi$,
the key idea underlying CAD is to decompose the state space $\Real^n$ into a finite number of regions called \emph{cells}.
This decomposition satisfies two primary properties: (i) in each cell each polynomial in $P$ remains a constant sign, and (ii) the boundary of each cell can be derived from polynomials in $P$.
Due to the first property, the sign of a polynomial in $P$ over a cell can be determined by a single sample point within that cell. Consequently, a finite set of sample points from the decomposition is sufficient to determine in which cells the conjunction of atomic formulas in $\Phi$ evaluates to true.
Finally, by utilizing the formulas that define these cells, we can construct a quantifier-free formula equivalent to $\Phi$.  

The CAD algorithm consists of two phases, \emph{projection} and \emph{lifting}, while in this paper, \emph{we solely focus on the projection phase.} 
The projection phase takes a variable ordering, say $\langle x_1, x_2 \dots, x_n\rangle$, where quantified variables precede free variables, and constructs a sequence of polynomial sets 
\begin{equation}
    \{P_1(=P), P_{2}, \dots, P_n\} \text{ with } P_i \subset \Real[x_i,\dots,x_n]   
\end{equation}
through the iterative application of a projection operator $\proj$.
Each subsequent set $P_{i+1} = \proj(P_i, x_i)$ is defined in a lower-dimensional space by removing the variable $x_{i}$ and carries enough information about $P_{i}$. 
There exist several different definitions of the projection operator, and in our analysis, we use McCallum's projection operator~\cite{mccallum88jsc}.
After the projection phase, the lifting phase then constructs the sign-invariant cells from 1-dimension to higher dimensions.
We provide a more detailed description of the CAD algorithm~\cite{arnon84jsc,jirstrand95,basu06book-rag} in~\cref{app:cad}.
which may not affect the understanding of this paper. 

\subsection{Treewidth and Tree Decomposition}
\label{ssec:graph}

In this part, we introduce basic concepts related to the sparsity of graphs. We recommend referring to \cite[Ch.~7]{cygan15book-parameterized} for an in-depth investigation of this topic.

\myparagraph{Graph and Primal Graph.}
We use $G=(V,E)$ to denote an undirected graph, where $V$ is a finite set of vertices and $E\subset V\times V$ is the set of edges.
The size of a graph, denoted $|G|$, is the cardinality of $V$.
A clique in a graph is a subset of mutually adjacent vertices.
Given a formula $\Phi$ in the form of \cref{eq:input}, the \emph{primal graph} $G_\Phi=(V_\Phi, E_\Phi)$ is a graph associated to $\Phi$ such that $V_{\Phi}=\{x_1,\dots,x_m\}$ is the collection of \emph{(existentially) quantified variables} and $(x_{j_1},x_{j_2})\in E_{\Phi}$ if and only if both $a_{i,j_1}\neq 0$ and $a_{i,j_2}\neq 0$ for some index $i$ of atomic formulas, i.e., $x_{j_1},x_{j_2}$ occur simultaneously in an atomic formula.
Throughout this paper, we will always assume the primal graphs are connected; otherwise, each connected component can be processed separately.

\begin{figure}[t!]
\vspace*{-2mm}
    \begin{minipage}{0.5\textwidth} 
        \centering
            \begin{align*}
                \Phi &\defeq \exists x_1,\dots, x_8.\\
                &
                \begin{pmatrix}
                \begin{aligned}
                    &~x_1 + 2 x_2 +3 x_3 \le 20 \\
                    \wedge &~x_1 - x_2 + 2 x_3 \ge -5\\
                    \wedge &~x_1 - 4 x_2 \le 0 \\
                    \wedge &~\dots \text{(constraints without $x_1$)}
                \end{aligned}   
                \end{pmatrix}
            \end{align*}
        \vspace*{-4mm} 
        \caption{An LRA formula $\Phi$}
        \label{fig:formula}
        \vspace*{2mm} 
        \par
        \resizebox{0.8\textwidth}{!}{
            \begin{tikzpicture}[node distance={15mm}, thick, main/.style = {draw, circle}] 
                \node[main] (1) at (0,1) {$x_1$}; 
                \node[main] (2) [below of=1] {$x_2$}; 
                \node[main] (3) [right of=1] {$x_3$}; 
                \node[main] (4) [right of=2] {$x_4$}; 
                \node[main] (5) [below of=2] {$x_5$}; 
                \node[main] (6) [right of=5] {$x_6$};
                \node[main] (7) [right of=6] {$x_7$};
                \node[main] (8) [left  of=5] {$x_8$};
                \draw (1) -- (3);
                \draw (1) -- (2);
                \draw (2) -- (3);
                \draw (2) -- (4);
                \draw (3) -- (4);
                \draw (2) -- (5);
                \draw (4) -- (5);
                \draw (4) -- (6);
                \draw (5) -- (6);
                \draw (6) -- (7);
                \draw (8) -- (2);
                \draw (8) -- (5);
            \end{tikzpicture} 
        }
        \captionof{figure}{The primal graph of $\Phi$}
        \label{fig:primal}
        \par
        \vspace*{2mm}
        \resizebox{0.55\textwidth}{!}{
            \begin{tikzpicture}[node distance={15mm}, thick, main/.style = {draw, rectangle}] 
            \node[main] (245) [fill = gray!40!white]{$2,4,5$}; 
            \node[main] (258) [below of=245, fill = gray!40!white] {$2,5,8$}; 
            \node[main] (234) [left  of=258, fill = gray!40!white] {$2,3,4$}; 
            \node[main] (123) [below of=234, fill = gray!40!white] {$1,2,3$};
            \node[main] (456) [right of=258, fill = gray!40!white] {$4,5,6$};
            \node[main] (67)  [below of=456, fill = gray!40!white] {$~6,7~$};
            \draw (245) -- (258);
            \draw (245) -- (234);
            \draw (245) -- (456);
            \draw (123) -- (234);
            \draw (456) -- (67);
            \end{tikzpicture} 
        }
            \captionof{figure}{A tree decomposition of {\cref{fig:primal}} }
            \label{fig:td}
        \end{minipage}
    \hfill
    \begin{minipage}{0.47\textwidth}
    \vspace*{4.5mm}
        \centering
        \resizebox{0.6\textwidth}{!}{
        \begin{tikzpicture}[node distance={10mm}, thick, main/.style = {draw, rectangle}] 
        \node[main] (root){$\emptyset$}; 
        \node[main] (2)   [below of=root]{$~~2~~$}; 
        \node[main] (24)  [below of=2]   {$2,4$}; 
        \node[main] (245) [below of=24, fill = gray!40!white]  {$2,4,5$}; 
        \node[main] (245-2) [below of =245]  {$2,4,5$}; 
        \node[main] (245-1) [left = 5mm of 245-2]   {$2,4,5$}; 
    
        \node[main] (24-2) [below of=245-1]   {$2,4$}; 
        \node[main] (234) [below of=24-2, fill = gray!40!white]   {$2,3,4$}; 
        \node[main] (23) [below of=234]   {$2,3$}; 
        \node[main] (123) [below of=23, fill = gray!40!white]   {$1,2,3$}; 
        \node[main] (12) [below of=123]   {$1,2$}; 
        \node[main] (1) [below of=12]   {$1$}; 
        
        \node[main] (245-3) [below of=245-2]   {$2,4,5$}; 
        \node[main] (245-4) [right=5mm of 245-3]   {$2,4,5$}; 
        \node[main] (25) [below of=245-3]   {$2,5$}; 
        \node[main] (258) [below of=25, fill = gray!40!white]   {$2,5,8$}; 
        \node[main] (28) [below of=258]   {$2,8$}; 
        \node[main] (8) [below of=28]   {$8$}; 
    
        \node[main] (45)  [below of=245-4]   {$4,5$}; 
        \node[main] (456) [below of=45, fill = gray!40!white]   {$4,5,6$}; 
        \node[main] (56)  [below of=456]   {$5,6$}; 
        \node[main] (6)   [below of=56]   {$6$}; 
        \node[main] (67)  [below of=6, fill = gray!40!white]   {$6,7$}; 
        \node[main] (7)   [below of=67]   {$7$}; 
        
        \draw (root) -- (2);
        \draw (2) -- (24);
        \draw (24) -- (245);
        \draw (245) -- (245-1);
        \draw (245) -- (245-2);
        \draw (245-1) -- (24-2);
        \draw (24-2) -- (234);
        \draw (234) -- (23);
        \draw (23) -- (123);
        \draw (123) -- (12);
        \draw (12) -- (1);
    
        \draw (245-2) -- (245-3);
        \draw (245-3) -- (25);
        \draw (25) -- (258);
        \draw (258) -- (28);
        \draw (28) -- (8);
    
        \draw (245-2) -- (245-4);
        \draw (245-4) -- (45);
        \draw (45) -- (456);
        \draw (456) -- (56);
        \draw (56) -- (6);
        \draw (6) -- (67);
        \draw (7) -- (67);
        
        \end{tikzpicture} 
        }
        \captionof{figure}{The nice tree decomposition of {\cref{fig:td}}. Gray nodes correspond to the nodes in {\cref{fig:td}}}
        \label{fig:nice-td}
    \end{minipage}
    \vspace*{-2mm}
\end{figure}

\begin{example}[Running Example]\label{ex:running}
Consider an LRA formula~$\Phi$ of the form~\cref{eq:input} and contains existentially quantified variables~$(x_1,\dots, x_8)$.
The portion of $\Phi$ related to $x_1$ is displayed in \cref{fig:formula}; the complete formula containing 20 atomic formulas is in the full version~\cref{app:example}.
The atomic formulas in $\Phi$ display a certain ``sparsity'' pattern---each variable only occurs in a small portion of atomic formulas.
For example, variable~$x_1$ only occurs in three atomic formulas, namely $x_1 + 2 x_2 +3 x_3 \le 20$, $x_1 - x_2 + 2 x_3 \ge -5$, and $x_1 - 4 x_2 \le 0$.
The primal graph of $\Phi$ is shown in \cref{fig:primal}.
\end{example}

\myparagraph{Tree.} 
A tree $T=(V,E)$ is an undirected graph that contains no cycles, i.e., there is only one unique path between any two vertices.
For clarity, the vertices of a tree will be called \emph{nodes}.
We always assume that a tree under consideration is \emph{rooted}, i.e., a node $r\in V$ is designated as the root.
The \emph{level} of a node $v\in V$, denoted $\textsf{Lv}(v)$, is the length of the simple path from $r$ to $v$. 
The \emph{height} of a tree is the maximum level among all nodes.
The \emph{parent} of a node $v$ is the node connected to $v$ on the path to the root, and a child of a node $v$ is a node of which $v$ is the parent.
A \emph{leaf} is a node with no children.
A tree is called \emph{binary} if each node has at most two children.
For each node $v$ in $T$, we denote by $T_v$ the subtree rooted at $v$; the nodes in $T_v$ (excluding $v$ itself) are called the \emph{descendants} of $v$.

\myparagraph{Tree Decomposition and Treewidth~\cite{halin76-sfunction,robertson84-minor3,robertson86-minor2}.} 
Given an undirected graph $G=(V,E)$, a \emph{tree decomposition} of $G$ is a tree $T=(\mathcal{B}, E_T)$ satisfying the following conditions:
(1) Each node $b\in \mathcal{B}$ of $T$, also called a \emph{bag}, is associated to a subset of the vertex set $V_b \subseteq V$;
(2) The bags cover the entire vertex set $V$ of $G$, i.e., $\bigcup_{b\in\mathcal{B}} V_b = V$;
(3) For every edge $e=(u,v)\in E$, there exists a bag $b$ such that both $u\in V_b$ and $v\in V_b$; and
(4) For every three bags $b_1,b_2,b_3\in \mathcal{B}$, if $b_3$ is on the path from $b_1$ to $b_2$ in $T$, then $V_{b_1}\cap V_{b_2} \subseteq V_{b_3}$ (equivalently, every vertex $v\in V$ appears in a connected subtree of $T$). 



The \emph{width} of a tree decomposition $T$, denoted by $w(T)$, is defined to be the size of the largest bag minus one, i.e.,  $w(T)=\max_{b\in \mathcal{B}} |V_b|-1$. 
The \emph{treewidth} of the graph $G$, denoted by $tw(G)$, is defined to be the minimum width among all possible tree decompositions of $G$. 
A tree decomposition with width $tw(G)$ is called \emph{optimal}. 
For example, \cref{fig:td} presents an optimal tree decomposition of the primal graph in \cref{ex:running} of width $2$.
The numbers in the labels refer to the variable indexes, for instance, the node labeled~$\{1,2,3\}$ in~\cref{fig:td} represents the bag~$\{x_1, x_2, x_3\}$.
For a primal graph $G_{\Phi}$, since every clique is contained in at least one bag in its tree decomposition~\cite{cygan15book-parameterized}, the largest number of quantified variables in an atomic formula gives a lower bound for $tw(G_\Phi)$. 

A smaller treewidth indicates a greater sparsity. 
In general, it is NP-complete to determine whether a given graph $G$ has treewidth at most a given variable~$k$~\cite{stefan87jadm}. 
However, Bodlaender showed that
when $k$ is fixed, there exists a linear-time algorithm to test whether $G$ has treewidth at most $k$, and if so, to construct an optimal tree-decomposition with $O(n)$ nodes~\cite{bodlaender96joc}.
In this paper, we rely on an external algorithm~\cite{tamaki19jco-td} to compute tree decompositions.
There are also many other well-optimized tools, see \cite{dell17ipec-pace} for a survey.

\myparagraph{Nice Tree Decomposition~\cite{kloks94book-treewidth}.}
A tree decomposition $T=(\mathcal{B}, E_T)$ of a graph $G$ is \emph{nice} if $T$ is rooted at a bag $r\in \mathcal{B}$ with $V_r=\emptyset$ and each bag of $T$ is of one of the following four types:
(1) \emph{leaf} bag is a bag $b$ having no children and $|V_b|=1$;
(2) A \emph{join} bag is a bag $b$ having exactly two children $b_1$ and $b_2$ such that $V_b=V_{b_1}=V_{b_2}$;
(3) An \emph{introduce} bag is a bag $b$ having exactly one child $b'$ such that $|V_{b}|=|V_{b'}|+1$ and $V_{b'}\subseteq V_{b}$;
(4) A \emph{forget} bag is a bag $b$ having exactly one child $b'$ such that $|V_{b}|=|V_{b'}|-1$ and $V_{b}\subseteq V_{b'}$.
As an example, \cref{fig:nice-td} gives a nice tree decomposition based on the tree decomposition in \cref{fig:td}. 


The concept of nice tree decomposition is introduced as a canonical form for representing tree decompositions.
When the treewidth is a constant, a tree decomposition can be transformed into a nice tree decomposition in linear time, with the same width but a linear increase in size~\cite[Lem.~7.4]{cygan15book-parameterized}, 
which helps design dynamic programming algorithms and conduct computational analysis. 

%% file: sections/3-algorithm.tex
\section{Exploiting Treewidth Sparsity in Quantifier Elimination}
\label{sec:algo}

Now we show how to exploit the treewidth sparsity of the primal graph in quantifier elimination. 
In \cref{ssec:fme}, we explain our main idea using \cref{ex:running} and present our dynamic programming framework for FME based on tree decomposition. 
Then we extend the framework to CAD and the general quantifier elimination procedures in \cref{ssec:extension}. 
The formal analysis of the complexity improvement brought by the treewidth sparsity will be presented in \cref{sec:complexity}.

\subsection{A Dynamic Programming Framework for FME}\label{ssec:fme}


To exploit the sparsity pattern within the primal graph, we begin by showing that the FME algorithm can be understood as a vertex elimination operation.
Let us consider applying the FME algorithm to our running example \cref{ex:running} to eliminate the variable $x_1$.
We first rewrite the three atomic formulas containing $x_1$ into the following forms, separating $x_1$ from other terms:
$x_1 \le 20 - 2 x_2 - 3 x_3 \wedge -5+x_2-2x_3 \le x_1 \wedge x_1 \le 4x_2$,
which give upper bounds and lower bounds for $x_1$. 
Then, the FME algorithm removes $x_1$ by requiring that each lower bound is less than each upper bound. 
This results in two atomic formulas
$-5+x_2-2x_3 \le 20 - 2 x_2 - 3 x_3 \wedge -5+x_2-2x_3 \le 4x_2$,
which only involve $x_2$ and $x_3$.
Here, a straightforward yet important observation is that the FME algorithm is performed \emph{locally}---when a variable is eliminated (e.g., $x_1$), the variables in the newly introduced atomic formulas will be exactly its neighbors in the primal graph (e.g., $x_2$ and $x_3$).
In other words, eliminating variables from a formula corresponds to removing vertices from the primal graph, accompanied by the introduction of new atomic formulas about the neighbor.



Using tree decompositions, we can re-explain the process of eliminating $x_1$ in a tree-like manner. 
Consider the nice tree decomposition presented in \cref{fig:nice-td},
the process of eliminating $x_1$ consists of two steps.
In the first step, we assign each atomic formula of $\Phi$ to a bag in the nice tree decomposition depending on the involved variables. 
This is described by the initialization of a function $\mathcal{I}$ that maps each bag to a set of atomic formulas.
For example, for the three bags located at the left-most branch of the nice tree decomposition in~\cref{fig:nice-td}, we set
$\mathcal{I}(\{x_1\}) \gets \{\}$ (meaning there are no atomic formulas containing only $x_1$),
$\mathcal{I}(\{x_1,x_2\}) \gets \{x_1-4x_2\le 0\}$,
and $\mathcal{I}(\{x_1,x_2,x_3\}) \gets \{x_1-x_2+2x_3\ge -5, x_1 + 2 x_2 +3 x_3 \le 20\}$.
Note that the atomic formula $x_1-4x_2\le 0$ can also be assigned to $\mathcal{I}(\{x_1,x_2,x_3\})$ without affecting the result.

In the second step, we propagate the atomic formulas stored in $\mathcal{I}$ from bottom up, by recursively computing another function $\mathcal{V}$. For example, for the leaf bag $\{x_1\}$, we set $\mathcal{V}(\{x_1\})=\mathcal{I}(\{x_1\})=\{\}$. 
For the introduce  bag $\{x_1,x_2\}$, we set $\mathcal{V}(\{x_1,x_2\}) \gets \mathcal{I}(\{x_1,x_2\})\cup \mathcal{V}(\{x_1\})=\{x_1-4x_2\le 0\}$.
The computation of $\mathcal{V}(\{x_1,x_2,x_3\})$ is similar.
When the forget bag $\{x_2, x_3\}$ is reached, we set $\mathcal{V}(\{x_2,x_3\}) \gets \mathcal{I}(\{x_2,x_3\}) \cup \textsf{FME}(\bigwedge \mathcal{V}(\{x_1,x_2,x_3\}), x_1)$, which corresponds an application of the FME algorithm w.r.t. the variable removed at this bag.


By repeating the above operations, the execution of the FME algorithm can be described through a dynamic programming framework, as shown in \cref{algo:fme}, which consists of three steps.

\begin{algorithm2e}[!t]
\caption{A Dynamic Programming Framework for FME}
\label{algo:fme}
\SetKwInOut{Input}{Input}
\SetKwInOut{Output}{Output}
\Input{a LRA formula $\Phi$ of the form \cref{eq:input};}
\Output{a quantifier-free formula equivalent to $\Phi$}
\BlankLine
\algocomment{Invoke external algorithms to compute nice tree decompositions.}
Construct the primal graph~$G_\Phi=(V_\Phi,E_\Phi)$ of $\Phi$\;
Compute a nice tree decomposition~$T=(\mathcal{B},E_T)$ of~$G_\Phi$\; 
\algocomment{Initialize $\mathcal{I}$ from top down.}
$C\gets$ atomic formulas in $\Phi$\;
Initialize an empty queue of bags $q=\langle \rangle$\;
$q.\texttt{enqueue}(r_T)$ \tcp*{add to the end of the queue}
\While{$q$ is not empty}{
    $b \gets q.\texttt{dequeue}()$ \tcp*{remove from the front of the queue}
    $q.\texttt{enqueue}(b.\texttt{children})$\tcp*{break ties with any heuristics}
    \eIf{$b.\textnormal{\texttt{type}} = \textsf{leaf}$ }{
        $\mathcal{I}(b) \gets \{\}$ \tcp*{if $b$ is a leaf bag}
    }{
        $\mathcal{I}(b) \gets \{\varphi\in C \mid \texttt{var}(\varphi)\subseteq b\}$ \tcp*{if $b$ is of other types}
        $C \gets C\setminus \mathcal{I}(b)$\;
    }
}
\algocomment{Recursively compute $\mathcal{V}$ from bottom up.}
$q.\texttt{enqueue}(\mathcal{B}.\texttt{leaves})$ \;
\While{$q$ is not empty}{
    $b \gets q.\texttt{dequeue}()$\;
    $q.\texttt{enqueue}(b.\texttt{parent})$\;
    Compute $\mathcal{V}(b)$ according to \cref{eq:FME-update}\;
}
\Return{$\bigwedge \mathcal{V}(r_T)$}
\end{algorithm2e}

\textbf{Step 1}: Constructing a nice tree decomposition (Line 1--3).
To exploit the sparsity of the input formula, the algorithm first computes a tree decomposition of the formula's primal graph, by invoking some external tree decomposition algorithms. 
Different choices of the algorithms lead to different tree decompositions and hence impact the overall computational complexity, which will be detailed later in \cref{sec:complexity} and \cref{sec:exp}.
Nevertheless, our framework works as long as $T=(\mathcal{B},E_T)$ in Line 2 is a nice tree decomposition.

\textbf{Step 2}: Initialization (Line 4--17).
Based on the nice tree decomposition~$T$, we can perform the FME algorithm in a dynamic programming style.
To achieve this, we define two value functions, $\mathcal{I}$ and $\mathcal{V}$, that map bags in $T$ into sets of atomic formulas. 
In this step, the value function $\mathcal{I}$ is initialized from top down: when a bag $b\in T$ is visited, all atomic formulas $\varphi$ with $\textsf{var}(\varphi)\subseteq b$ are collected in $\mathcal{I}(b)$ and removed from $C$. 
According to the definition of tree decompositions, it is straightforward to see that the following properties hold:
Each atomic formula~$\varphi\in C$ is assigned to a bag $b$ such that $\textsf{var}(\varphi)\subseteq b$ and $\textsf{Lv}(b)$ is minimized. We note that, when adding the children of a bag $b$ into the queue in Line~10, we can break ties with any user-specified heuristics. The selection of the heuristics impacts the overall efficiency in the practical implementation, which will be detailed in \cref{sec:exp}.

\textbf{Step 3}: Recursion (Line 18--25).
The value function $\mathcal{V}$ is then updated from bottom up, i.e., $\mathcal{V}(b)$ is updated after the values for all $b$'s children have been updated.
Let us denote by $b.\texttt{(left/right)child}$ the (left/right) child of a (join) bag $b$, 
by $b.\texttt{children}$ the collection of all children of a bag $b$,
and by $b.\texttt{forget}$ the removed variable at a forget bag $b$.
At each step, $\mathcal{V}(b)$ is updated according to the type of the bag $b\in T$:
\begin{equation}\label{eq:FME-update}
\begin{aligned}
    \mathcal{V}(b) \gets 
    \begin{cases}
        \mathcal{I}(b)\cup \mathcal{V}(b.\texttt{leftchild}) \cup \mathcal{V}(b.\texttt{rightchild}) &\text{ if $b$ is a join bag}\\
        \mathcal{I}(b)\cup \mathcal{V}(b.\texttt{children}) &\text{ if $b$ is an introduce bag}\\
        \mathcal{I}(b)\cup \textsf{FME} (\bigwedge\mathcal{V}(b.\texttt{children}), b.\texttt{forget}) &\text{ if $b$ is a forget bag}\\
        \mathcal{I}(b) &\text{ if $b$ is a leaf bag} 
    \end{cases}
\end{aligned}    
\end{equation}

Finally, the algorithm outputs the conjunction of all atomic formulas stored in $\mathcal{V}(r_T)$.

\begin{theorem}[Correctness]
    The formula $\bigwedge \mathcal{V}(r_T)$ is equivalent to $\Phi$.
\end{theorem}

\begin{proof}
    Let $T_{b}=(\mathcal{B}_b, E_{T_b})$ denote the sub-tree rooted at a bag $b$. 
    We denote by $\textsf{ElimVar}(T_b)\defeq (\bigcup_{b'\in T_b} b') \setminus b$ the set of removed variables at forget bags in $T_b$.
    In this proof, we prove a stronger result by showing that, for each bag $b\in \mathcal{B}$, 
    $\bigwedge \mathcal{V}(b)\equiv \exists \textsf{ElimVar}(T_b).~\bigwedge_{b'\in \mathcal{B}_b} \mathcal{I}(b')$.
    Therefore, the original statement is a special case when $b$ is the root $r$ of $T$, by noting that $\Phi=\exists x_m, \dots, \exists x_1.~\bigwedge_{b\in \mathcal{B}} \mathcal{I}(b)$. 
    
    The proof is done by induction on the structure of the nice tree decomposition~$T$. 
    Suppose the above property holds for every descendant of a bag $b$, we shall prove that the property still holds for $b$, depending on the type of $b$.
    When $b$ is a join, introduce, or leaf bag, the proof is straightforward by the definition of $\mathcal{V}$.
    When $b$ is a forget bag, we have 
    \begin{equation*}
        \begin{aligned}
            \bigwedge \mathcal{V}(b) &= \bigwedge \big( \mathcal{I}(b)\cup \textsf{FME} (\bigwedge\mathcal{V}(b.\texttt{children}), b.\texttt{forget})\big)\\
            & \equiv \bigwedge \mathcal{I}(b) \wedge \bigwedge  \textsf{FME} (\exists \textsf{ElimVar}(T_{b.\texttt{children}}).\bigwedge_{b'\in \mathcal{B}_{b.\texttt{children}}} \mathcal{I}(b'), b.\texttt{forget}) \\
            & \equiv \bigwedge \mathcal{I}(b) \wedge \bigg( \exists \textsf{ElimVar}(T_{b.\texttt{children}}) \exists b.\texttt{forget}.~\bigwedge_{b'\in \mathcal{B}_{b.\texttt{children}}} \mathcal{I}(b')\bigg)\\
            &\equiv \exists \textsf{ElimVar}(T_b).~\bigwedge \mathcal{I}(b) \wedge\bigwedge_{b'\in \mathcal{B}_{b.\texttt{children}}} \mathcal{I}(b') 
            \equiv \exists \textsf{ElimVar}(T_b).~\bigwedge_{b'\in \mathcal{B}_b} \mathcal{I}(b'),
        \end{aligned}
    \end{equation*}
    where the second line is obtained by applying the induction hypothesis, and the last line is correct because variables in $\textsf{ElimVar}(T_b)$ do not occur in $\mathcal{I}(b)$.
\end{proof}


\subsection{Extending to CAD and General QE}\label{ssec:extension}

\myparagraph{Extending to CAD's projection phase.}
The framework~\cref{algo:fme} can be utilized in CAD's projection phase with only minor modifications.
We present the whole algorithm in the full version~\cref{app:algo}
and only explain the differences here:
(i) The two value functions $\mathcal{I}$ and $\mathcal{V}$ are defined as mappings from bags in $T$ to set of polynomials in $\Real[x_1,\dots,x_n]$, instead of atomic formulas;
(ii) In the recursion step, for a forget bag $b$, we compute  $\mathcal{V}(b)$ via $\mathcal{V}(b) \gets \mathcal{I}(b)\cup \textsf{Proj} (\mathcal{V}(b.\texttt{children}), b.\texttt{forget})$,
while for other types of bags, the computation of $\mathcal{V}(b)$ remains the same as in \cref{algo:fme}; and
(iii) Finally, after applying the projection operation for all quantified variables, the CAD algorithm continues the projection phases for free variables, after which the lifting phase can start.
The correctness of the framework follows from the fact that the algebraic operations in the projection operator (presented in~\cref{app:cad}) preserve locality, i.e., will not introduce new variables.
We shall note here that we \emph{cannot} simply replace the invocation of the FME algorithm in \cref{algo:fme} by the CAD algorithm, i.e., the projection phase and lifting phase of CAD must be separated. 
This is because the lifting phase produces a formula in disjunction normal form, on which we can not directly apply the dynamic programming framework.

\myparagraph{Extending to general quantifier elimination algorithms.} 
We now extend our framework to an arbitrary logic theory $\mathcal{T}$ that admits quantifier elimination.
This is achieved by extracting a variable elimination order from the tree decomposition, which involves the following steps:
(1) Construct a nice tree decomposition $T$ of the input $\mathcal{T}$-formula's primal graph;
(2) Traverse $T$ from top to bottom to obtain an ordering of the eliminated variables corresponding to the forget bags; and
(3) Finally, reverse the ordering and use it as the variable elimination order in any quantifier elimination procedure for theory $\mathcal{T}$.
The pseudo-code description of the algorithm is presented in~\cref{app:algo}.
The correctness of the algorithm is straightforward, as the change of variable order will not affect the correctness of a quantifier elimination algorithm.

In fact, the variable ordering extracted via the above algorithm can be shown to be a \emph{perfect elimination ordering}\footnote{In graph theory, a perfect elimination ordering of a graph $G$ is an ordering of its vertices such that for every vertex $v$, the set consisting of $v$ and its neighbors that come after $v$ in the ordering forms a clique in $G$.}, say $(1,3,8,7,6,5,4,2)$ for our running example.
Hence, it is possible that our variable ordering coincides with the one extracted from a chordal extension graph~\cite{li23jsc}. 
However, since we have exploited the sparsity information in building the tree decompositions, our variable orderings are usually superior for performing quantifier elimination on sparse problem instances.
This will be further evidenced by our experimental results in \cref{sec:exp}.
Moreover, our framework is stated for general quantifier elimination procedures and enables a fine-grained complexity analysis, explained in the following \cref{sec:complexity}.

%% file: sections/4-analysis.tex
\section{Complexity Analysis}\label{sec:complexity}

In this section, we analyze the worst-case complexity of our framework in \cref{sec:algo}.
The core idea is to combine the standard analysis with a special kind of tree decomposition, called \emph{balanced tree-decomposition}. 
We prove that when the treewidth is constant, a nice and balanced tree-decomposition with logarithmic height in the number of variables can be found in sub-polynomial time.
By working on this tree decomposition, our dynamic programming framework can achieve exponential improvement in the worst-case complexity for both FME and CAD.
To our knowledge, this gives the first complexity result of FME and CAD directly related to treewidth sparsity.

\myparagraph{Balanced tree-decomposition.}
The concept of \emph{$(\beta,\gamma)$-balanced tree decompositions}~\cite{chatterjee18toplas} generalizes the notion of balanced tree decompositions that arise in the analysis of the space complexity of tree decomposition algorithms~\cite{bodlaender98joc,elberfeld10focs}.
Informally, a $(\beta,\gamma)$-balanced tree-decomposition requires, for each bag $b$, the size of subtree $T_b$ decreases proportionally as $\textsf{Lv}(b)$ increases.

\begin{definition}[$(\beta, \gamma)$-Balanced Tree Decomposition] 
For constants $0<\beta<1$ and $\gamma\in \Nat_{>0}$,
a binary tree-decomposition $T$ is called $(\beta, \gamma)$-\emph{balanced} if, for every bag $b$, every descendant $b'$ of $b$ with $\textsf{Lv}(b')-\textsf{Lv}(b)=\gamma$ satisfies that $|T_{b'}|\le \beta \cdot |T_b|$.
\end{definition}

The construction of a $(\beta, \gamma)$-balanced tree decomposition is elucidated in~\cite[Sec.~3]{chatterjee18toplas}, which gives the following theorem.

\begin{theorem}[{\cite[Thm.~3.1]{chatterjee18toplas}}]\label{thm:balanced}
Given a graph $G$ with $n$ vertices and of constant treewidth~$k$.
For any fixed $\delta>0$ and $\lambda\in \Nat$ with $\lambda\ge2$, a binary $(\beta, \gamma)$-balanced tree-decomposition $T$ with $\beta = (\frac{1+\delta}{2})^{\lambda-1}$ and $\gamma = \lambda$ can be constructed in $O(n\cdot \log n)$ time and $O(n)$ space.
Moreover, the tree decomposition $T$ has $O(n)$ bags and width $\frac{6\lambda}{\delta}\cdot(k+1)-1$  at most.
\end{theorem}

Relying on the above theorem, we can construct a nice tree decomposition with logarithmic height, as stated below. 

\begin{proposition}\label{prop:balanced}
    Given a graph $G$ with $n$ vertices and of constant treewidth $k$, 
    a nice tree decomposition $T$ of height $O(\log n)$ and constant width can be constructed in $O(n\log n)$ time.
\end{proposition}
\begin{proof}
    Take $\delta=1/2$ and $\lambda=2$ in \cref{thm:balanced}, then a $(3/4, 2)$-balanced tree decomposition $T=(\mathcal{B},E_T)$ of width at most $24(k+1)-1$ can be constructed in $O(n\cdot \log n)$ time.
    Because of the $(3/4,2)$-balancing, we have in every 2-levels of $T$ a decrease of at least 3/4 in the size of the bags.
    Therefore, the height of $T$ is $2\log_{\frac{3}{4}} n$, which is $O(\log n)$. Then, we extend $T$ into a nice tree decomposition, which involves the following steps:

    Step 1: reroot. Choose an arbitrary bag $r\in \mathcal{B}$ of $T$ to be its current root. We first construct a new bag $r'$ such that $V_{r'}=\emptyset$ as the new root of the nice tree decomposition. Then, a sequence of forget bags is introduced to remove the variables in $V_r$ one by one until reaching $r'$. This step takes $O(k)$ time and increases the height of $T$ by $k$ at most.

    Step 2: ensure join bags have identical children. For each bag $b$ that has two children, we construct two more bags that are identical to $b$ and insert them between $b$ and its two children, respectively. This step takes $O(n)$ time and will double the height at most.

    Step 3: create introduce and forget chains. After the previous steps, for each edge $(b_1, b_2)$ in $T$, the bags $V_{b_1}$ and $V_{b_2}$ may still differ by more than one variables. So we insert between $b_1$ and $b_2$ a chain of introduce and forget bags to make sure that every consecutive pair of bags differ in at most one variable. This step takes $O(kn)$ time and increases the height by up to $k$ times at most.

    To conclude, when $k$ is a constant, we can obtain a nice tree decomposition of height $O(\log n)$ in $O(n)$ time. Hence, the total time cost remains $O(n\log n)$ as in \cref{thm:balanced}.
\end{proof}

In the following, we analyze the complexity of our frameworks for FME and CAD in~\cref{sec:algo} under the assumption that $k$ is a constant. 
For consistency\footnote{Recall that the CAD's projection phase requires performing the projection operation for all variables, including the quantifier-free ones.}, we assume that $m=n$ in the input formula~\cref{eq:input}, i.e., there are no free variables.

\myparagraph{Complexity analysis for FME.} 
The complexity of the FME algorithm can be measured by the number of atomic formulas produced in the elimination process.
Assume that there are $s$ atomic formulas in the input formula~\cref{eq:input}.
Then, after eliminating the variable $x_r$, the number of atomic formulas in the output depends on the number of lower/upper bounds, which is at most $|L_r|\cdot|U_r|\le \frac{s^2}{4}$.
Therefore, the complexity of eliminating $n$ variables is at most $O(s^{2^n})$.

By combining the standard analysis with the nice tree decomposition obtained in \cref{prop:balanced}, we show that a single exponential complexity upper bound w.r.t. $n$ can be achieved when the treewidth is a constant.

\begin{theorem}\label{thm:fme-complexity}
    If the primal graph of an LRA formula $\Phi$ in \cref{eq:input} has constant treewidth and includes $s$ atomic formulas, the number of atomic formulas in the output of \cref{algo:fme} is $s^{n^{O(1)}}$ at most.
\end{theorem}
\begin{proof}
    Let $T$ be the nice tree decomposition given by \cref{prop:balanced} with height $h$. 
    In the following, we count the atomic formulas stored in $\mathcal{V}(b)$ for each bag $b$ in $T$.
    For a bag $b$ with $\textsf{Lv}(b)=h$, which must be a leaf bag, we have $\mathcal{V}(b)=\mathcal{I}(b)\le s$.
    Consider going one step up, for a bag $b$ with $\textsf{Lv}(b)=h-1$: 
    if $b$ is a join, introduce, or leaf bag, we have $\mathcal{V}(b)\le 3s$ according to the definition of $\mathcal{V}$;
    if $b$ is a forget bag, we apply the one-step analysis for FME and have $\mathcal{V}(b)\le s+\frac{s^2}{4}$.
    Since the height $h$ is bounded by $O(\log n)$ according to \cref{prop:balanced}, for the root $r$, we have $\mathcal{V}(r)=O(s^{2^{O(\log n)}})=s^{n^{O(1)}}$.
\end{proof}

\myparagraph{Complexity analysis for CAD.}
The CAD complexity is usually measured by the number of sign-invariant cells.
Our analysis uses a similar approach as \cite{li23jsc} but provides an improved bound based on the assumption of constant treewidth and the structure of balanced tree decompositions.

Before presenting the results, we introduce some tools to estimate the growth of the size of the polynomials in the projection phase.
For a set of polynomials $P\subset \Real[\seq{x}]$, we define the \emph{combined degree}~\cite{mccallum85-thesis} of $P$ to be $\max_{1\le i\le n} \deg_{x_i}(\prod_{p\in P} p)$, where $\deg_{x_i}(\cdot)$ denotes the degree of a polynomial in a single variable $x_i$. 
A set of polynomials has the $(m, d)$-property if it can be partitioned into $m$ sets, such that each set has maximum combined degree $d$~\cite[Def.~7]{bradford16jsc}.
It has been shown in \cite[Lem.~11]{bradford16jsc} that for a set $P\subset \Real[\seq{x}]$ with the $(m, d)$-property, the set $\textsf{Proj}(P,x_i)$ after projection has the $(M, 2d^2)$-property with $M=\lfloor \frac{(m+1)^2}{2}\rfloor $ for any $x_i\in \seq{x}$. 

\begin{proposition}
    If the primal graph of an NRA formula $\Phi$ in \cref{eq:input} has constant treewidth and the corresponding polynomial set has the $(m,d)$-property, the number of cells in the CAD algorithm can be made $m^{n^{O(1)}}(2Md)^{n^{O(1)}}$ at most.
\end{proposition}

\begin{proof}
    The proof is inspired by the proof of \cite[Thm.~11, Thm.~12]{li23jsc}. Here we explain the main idea. 
    
    We first analyze the increase in the number of polynomials in the projection phase.
    Let $T$ be the nice tree decomposition given by \cref{prop:balanced} with height $h$. 
    By assumption, for each leaf bag $b$ with $\textsf{Lv}(b)=h$, the polynomial set $\mathcal{V}(b)$ has the $(m,d)$-property.
    Following a similar argument in the proof of \cref{thm:fme-complexity}, we can show that for each bag $b$ with $\textsf{Lv}(b)\le h-1$, the polynomial set $\mathcal{V}(b)$ has the $\big( 3^{2^{l_b}-1}M^{2^{{l_b}-1}},2^{2^{l_b}-1}d^{2^{l_b}}\big)$-property with $l_b=h-1-\textsf{Lv}(b)$. 

    Then we analyze the lifting phase. 
    Note that, for a polynomial set $P$, the real roots of the product $\prod_{p\in P} p$ include all real roots of each individual polynomial in $P$.
    Hence, if a univariate polynomial set $P$ has the $(m,d)$-property, the number of real roots in $P$ is at most $md$, and the number of the corresponding sets in $\Real^1$ is at most $2md+1$. 
    The total number of cells produced in the lifting phase is bounded by the product of all $\Real^1$-cells corresponding to each bag $b$, i.e. $\prod_{b\in \mathcal{B}} (2K_b+1)$, where
    \begin{equation*}
        K_b = \begin{cases}
            md, &\quad \text{if $\textsf{Lv}(b)=h$, i.e., $l_b$ undefined;}\\
            3^{2^{l_b}-1}M^{2^{{l_b}-1}}2^{2^{l_b}-1}d^{2^{l_b}}, &\quad \text{if $h-1\ge \textsf{Lv}(b)\ge 0$, i.e., $0\le l_b\le h-1$.}
        \end{cases}
    \end{equation*}
    Since $T$ is a binary tree, the number of bags on each level is bounded by $2^{\textsf{Lv}(b)}$. 
    Therefore, we have
    \begin{equation*}
    \small{
    \begin{aligned}
        \prod_{b\in \mathcal{B}} (2K_b+1) &\le (md)^{2^h}\prod_{0\le l_b\le h-1} \big( 3^{2^{l_b}-1}M^{2^{{l_b}-1}}2^{2^{l_b}-1}d^{2^{l_b}}\big)^{2^{h-1-l_b}}\\
        &\le m^{2^h} 6^{(h-2)2^{h-1}+1} M^{h2^{h-2}}d^{(h+2)2^{h-1}}
        = O(m^{2^h} 6^{h2^{h-1}}M^{h2^{h-2}}d^{h2^{h-1}})\\
        & = m^{n^{O(1)}}(2Md)^{n^{O(1)}}
    \end{aligned}
    }
    \end{equation*}
    where the final step is obtained by using $h=O(\log n)$.
\end{proof}

%% file: sections/5-experiments.tex
\section{Implementation and Experiments}\label{sec:exp}

In this section, we explain our implementation details and present the experimental results.
We aim to answer the following research question: 
Compared with the existing heuristics for deciding the variable elimination ordering, \emph{how does our tree-decomposition-based algorithm perform on the QE problems of LRA and NRA instances with relatively small treewidth?}


\myparagraph{Computation Environment.} 
We use the tool from~\cite{tamaki19jco-td} to compute tree decompositions, which are then processed by Python scripts.
For LRA, we compare different heuristics based on a naive FME implementation in Python, while for NRA, we work on CAD implemented in \textsc{Mathematica}~\cite{mathematica}.
All experiments were conducted on a Windows PC with i7-13700 CPU and 32GB RAM.

\myparagraph{Existing Heuristics.}
For comparison, we briefly explain the most frequently used heuristic strategies for selecting the variable elimination ordering in FME and CAD.
For FME, the baseline Random strategy makes $N=5$ trials according to different orderings that are sampled randomly and records the best result; the Greedy strategy chooses the next variable to be eliminated by minimizing the number of generated constraints.
For CAD, the Brown's heuristic, denoted Brown~\cite{brown04issac}, chooses the next eliminated variable that has the lowest maximum degree across all polynomials in which it appears (and breaks ties with addition rules);
the PEO strategy~\cite{li23jsc} extracts a perfect elimination ordering from the chordal extension of the primal graph;
and we also compare with the default strategy in \textsc{Maple}~\cite{maple} to suggest variable ordering, denoted SVO.

\begin{table}[t]
    \centering
    \caption{FME performance on randomly generated LRA instances}
    \label{tab:fme}
    \begin{tabular}{c|c r|c|c|c}
        \hline
        ID & \multicolumn{2}{c|}{Instance Set (10 instances per set)} & Random & Greedy & Ours \\
        \hline

        \multirow{2}{*}{1} & \multirow{2}{*}{\makecell{\#var=15\  \#ineq=75\\ \#elim=5\ \#tw=2}} & Ave. ineq. count & 51,608.4 & \textbf{4,404.7} & 4,841.3 \\
        \cline{4-6}
        & & Ave. run~~Time(s) & 0.21 & \textbf{0.018} & \textbf{0.018} \\
        \hline

        \multirow{2}{*}{2} & \multirow{2}{*}{\makecell{\#var=15\  \#ineq=75\\ \#elim=5\ \#tw=4}} & Ave. ineq. count & $>$10,000,000 & 28,362.3 & \textbf{28,244.8} \\
        \cline{4-6}
        & & Ave. run~~Time(s) & NA & \textbf{0.17} & \textbf{0.17} \\
        \hline
        
        \multirow{2}{*}{3} & \multirow{2}{*}{\makecell{\#var=20\  \#ineq=100\\ \#elim=10\ \#tw=2}} & Ave. ineq. count & $>$10,000,000 & 290,597.8 & \textbf{194,645.3} \\
        \cline{4-6}
        & & Ave. run~~Time(s) & NA & 2.0 & \textbf{1.1} \\
        \hline

        \multirow{2}{*}{4} & \multirow{2}{*}{\makecell{\#var=20\  \#ineq=100\\ \#elim=6\ \#tw=4}} & Ave. ineq. count & $>$10,000,000 & \textbf{313,584.0} & 313,585.9 \\
        \cline{4-6}
        & & Ave. run~~Time(s) & NA & 2.6 & \textbf{2.4}  \\
        \hline

        \multirow{2}{*}{5} & \multirow{2}{*}{\makecell{\#var=30\  \#ineq=150\\ \#elim=12\ \#tw=3}} & Ave. ineq. count & $>$10,000,000 & 343,705.4 & \textbf{110,304.9} \\
        \cline{4-6}
        & & Ave. run~~Time(s) & NA & 3.7 & \textbf{1.2} \\
        \hline
        
        \multirow{2}{*}{6} & \multirow{2}{*}{\makecell{\#var=30\  \#ineq=150\\ \#elim=5\ \#tw=8}} & Ave. ineq. count & $>$10,000,000 & \textbf{974,233.4} & \textbf{974,233.4} \\
        \cline{4-6}
        & & Ave. run~~Time(s) & NA & 14.9 & \textbf{14.4} \\
        \hline
    \end{tabular}
\end{table}

\myparagraph{Our Implementation.}
In our practical implementation, we elected not to construct a balanced tree decomposition as described in \cref{sec:complexity}. While the main advantage of the balanced tree decomposition lies in estimating the asymptotic complexity, it often performs suboptimally on problems of relatively small scale. Instead, we work with the tree decomposition computed by \cite{tamaki19jco-td} and follow the approach detailed in \cref{ssec:extension} to extract the elimination order\footnote{In our non-parallel setting, using the elimination order is functionally equivalent to running the dynamic programming framework, without the explicit computation of the value functions.}. When traversing from top to bottom, we break ties using the aforementioned Greedy and Brown heuristics for FME and CAD, respectively.

\myparagraph{Benchmarks.}
As our algorithm targets dealing with LRA and NRA problems that are \emph{small enough} to do quantifier elimination but also \emph{have low treewidth}.
This requirement rules out most existing standard benchmarks (an analysis on SMT-LIB benchmarks~\cite{barrett16smtlib} can be found in the full version~\cref{app:smtlib}.
Hence, for comparison, we generated a set of benchmarks with treewidth sparsity as follows:
We construct benchmark inequality systems from graphs of prescribed treewidth $k$ by a two-stage procedure. We first generate a graph of treewidth $k$ by iteratively attaching nodes to a $(k+1)$-clique. Then, for each bag, we enumerate all monomials of total degree not exceeding the maximum degree ($1$ for LRA), select at least one highest-degree monomial, and randomly include other monomials with a probability ranging from 0.05 to 0.15, while ensuring that every scoped variable appears.
For each set of hyperparameters, we randomly generated 10 examples for LRA and 1 example for NRA, because the randomly selected monomials significantly impact the CAD efficiency, and the average runtime is usually dominated by the worst instance. 


 
\myparagraph{Results.} \cref{tab:fme} and \cref{tab:cad} present the results of FME and CAD, respectively. 
On LRA benchmarks, our strategy consistently yielded the lowest average inequality count and average runtime, achieving the best result in 4 out of 6 sets. 
The Greedy strategy showed a very similar performance with the runtime difference being negligible, but was significantly slower on ID 3 and 5.
The Random strategy performed poorly, generating an extremely high number of inequalities, often exceeding $10^7$ for most instances.
On NRA benchmarks, our strategy also demonstrates superior performance in all but one benchmark, in terms of both the number of cells generated and the runtime.
For benchmark ID 9 and ID 12, it significantly outperforms all other methods.
The only instance where our strategy did not have the lowest cell count was ID 10, where PEO was superior.
The results answer our research question affirmatively, indicating that \emph{our proposed strategy is effective and more efficient in dealing with problems with low treewidth}, compared to the existing heuristics. More experiments on benchmarks from~\cite{li23jsc} are available at~\cref{app:experiments}.

\begin{table}[t]
    \centering
    \caption{CAD performance on randomly generated NRA instances}
    \label{tab:cad}
    \begin{tabular}{c|c r|c|c|c|c}
        \hline
        ID & \multicolumn{2}{c|}{Instance} & SVO & Brown & PEO & Ours  \\
        \hline


        \multirow{2}{*}{7} & \multirow{2}{*}{\makecell{\#var=6\  \#ineq=6\\ \#max\_deg=2\ \#tw=2}} & Cells & 1,133,532 & 271,120 & 274,724 & \textbf{110,828}\\
        \cline{4-7}
        & & ~~Time(s) & 640.3 & 307.0 & 114.7 & \textbf{40.5} \\
        \hline



        \multirow{2}{*}{8} & \multirow{2}{*}{\makecell{\#var=7\  \#ineq=4\\ \#max\_deg=3\ \#tw=3}} & Cells & 126,328 & 89,804 & 89,804 & \textbf{49,108} \\
        \cline{4-7}
        & & ~~Time(s) & 42.1 & 34.5 & 29.0 & \textbf{11.6} \\
        \hline

        

        \multirow{2}{*}{9} & \multirow{2}{*}{\makecell{\#var=7\  \#ineq=5\\ \#max\_deg=2\ \#tw=2}} & Cells & 2,020,378 & 1,926,208 & 740,380 & \textbf{181,090} \\
        \cline{4-7}
        & & ~~Time(s) & 702.0 & 416.2 & 96.1 & \textbf{84.1} \\
        \hline


        \multirow{2}{*}{10} & \multirow{2}{*}{\makecell{\#var=8\  \#ineq=5\\ \#max\_deg=4\ \#tw=2}} & Cells & 3,721,372 & 3,721,372 & \textbf{984,008} & 2,155,132 \\
        \cline{4-7}
        & & ~~Time(s) & 1843.8 & 1718.2 & \textbf{567.7} & 1090.3 \\
        \hline

        \multirow{2}{*}{11} & \multirow{2}{*}{\makecell{\#var=8\  \#ineq=6\\ \#max\_deg=2\ \#tw=3}} & Cells & \textbf{38,212} & 69,748& 69,748 & \textbf{38,212}  \\
        \cline{4-7}
        & & ~~Time(s) & \textbf{42.0} & 633.3& 427.9 & \textbf{42.0}  \\
        \hline


        \multirow{2}{*}{12} & \multirow{2}{*}{\makecell{\#var=9\  \#ineq=10\\ \#max\_deg=2\ \#tw=3}} & Cells & 1,804,224 & 1,804,224 & 954,432 & \textbf{52,992} \\
        \cline{4-7}
        & & ~~Time(s) & 742.3 & 746.7 & 156.7 & \textbf{83.3} \\
        \hline
        
    \end{tabular}
\end{table}

%% file: sections/6-summary.tex
\section{Conclusion}\label{sec:summary}
This paper introduces a high-level, treewidth-aware approach to reduce the computational burden of quantifier elimination procedures, especially FME and CAD, for problems with treewidth sparsity. 
By leveraging parameterized algorithm tools, our method establishes an improved complexity upper bound, thus offering novel insights into mitigating the inherent complexity of quantifier elimination. 
Future work will focus on combining this treewidth-aware method with existing heuristics and exploring its application in various practical scenarios.

\begin{credits}
\myparagraph{\ackname} 
We thank the anonymous reviewers for their valuable comments and helpful suggestions.
This work has been partially funded by the National Key R\&D Program of China under grant No.\ 2022YFA1005101 and 2022YFA1005102, the National NSF of China under grant No.62192732 and W2511064,
the ERC Starting Grant 101222524 (SPES), and the Ethereum Foundation Research Grant FY24-1793.

\myparagraph{Data Availability Statement.} 
The code for our experiments is available at \url{https://doi.org/10.5281/zenodo.18150640}.

\end{credits}

%% file: sections/appendix.tex
\section{Appendix}

\subsection{CAD Algorithm}\label{app:cad}

In this part, we provide a formal description of the CAD algorithm.
The notations and definitions mostly follow Jirstrand's introductory paper~\cite{jirstrand95}.
Given a set of polynomials in $\mathcal{F}\subset \Real[x_1,\dots,x_n]$, the (CAD) algorithm decomposes the $\Real^n$ space into finitely many sign-invariant regions, called \emph{cells}. After decomposition, one can determine the sign of the polynomials in each cell using sampled points in this cell. 

In the one-dimensional case, the CAD algorithm can be understood as computing all the real roots of a set of univariate polynomials. Then, the sign-invariant cells are exactly the real roots themselves, along with the intervals they divide. Using this idea, $n$-dimensional CAD is constructed recursively based on $(n-1)$-dimensional CAD.

We first set some basic definitions. 
A region $R$ is a connected subset of $\Real^n$. 
The set $Z(R)=R\times \Real$ is called a \emph{cylinder} over $R$. 
Let $f,f_1,f_2$ be continuous, real-valued functions over $R$.
A $f$-section of $Z(R)$ is the set  $\{(x, f(x))\mid x\in R\}$.
A $(f_1,f_2)$-sector of $Z(R)$ is the set  $\{(x, y)\mid x\in R, f_1(x)< y < f_2(x)\}$.
A decomposition of $X$ is a finite collection of disjoint regions components whose union is $X$.

\begin{definition}[Stack]
    A \emph{stack} over $R$ is a decomposition which consists of $f_i$-sections and $(f_i,f_{i+1})$-sectors, where $f_0<\dots < f_{k+1}$ for all $x\in R$ and $f_0=-\infty$, $f_{k+1}=+\infty$.
\end{definition}

\begin{definition}[Delineable]
    A set of polynomials is \emph{delineable} over a region $R$ if the number of distinct collective roots remains a constant, i.e., their roots do not intersect or disappear.
\end{definition}

\begin{definition}[Cylindrical Decomposition]. A decomposition $\mathcal{D}$ of $\Real^n$ is \emph{cylindrical} if: 
(i) when $n=1$, $\mathcal{D}$ is a partition of $\Real$ into finitely set of numbers, and the finite and infinite open intervals bounded by these numbers.
(ii) when $n>1$, there exists a cylindrical decomposition $\mathcal{D}'$ of $\Real^{n-1}$, and over each region $R\in D'$, there is a stack belonging to $\mathcal{D}$.
\end{definition}

\begin{definition}[Cylindrical Algebraic Decomposition]
A decomposition is \emph{algebraic} if each of its components is a semi-algebraic set. 
A \emph{Cylindrical Algebraic Decomposition~(CAD)} of $\Real^n$ is a decomposition which is both cylindrical and algebraic. The components of a CAD are called \emph{cells}.
\end{definition}

\myparagraph{Defining the projection operator.}
In the following, we first recall the definition of resultants and discriminants for multivariate polynomials, then give the formal definition of McCallum's projection operator.
Let $\Int[x_1,\dots,x_n]$ denote the ring of all integral polynomials in $x_1,\dots,x_n$.
For any $x_k$ with $1\le k\le n$, a polynomial $f\in \Int[x_1,\dots,x_n]$ can be written as $f=\sum_{i=0}^d a_i x_k^i$, where the coefficients $a_i$ are polynomials in $\Real[x_1,\dots x_{k-1}, x_{k+1}, \dots, x_n]$ and $a_d$ is not the zero polynomial.
The number $d\in \Nat$ is called the \emph{degree} of $f$ in variable $x_k$, denoted $\textsf{deg}(f,x_k)$.
The set of non-zero coefficients is denoted by $\textsf{coeff}(f, x_k)$. 
The GCD of the coefficients is called the \emph{content} of $f$ w.r.t. $x_k$ denoted by $\textsf{cont}(f, x_k)$.
The polynomial $f$ is called \emph{primitive} w.r.t. $x_k$ if $\textsf{cont}(f, x_k)=1$, and the polynomial $f/\textsf{cont}(f, x_k)$ is called the \emph{primitive part} of $f$ w.r.t. $x_k$. 

\begin{definition}[Resultant]
Let $f_1, f_2$ be two polynomials in $\Real[x_1, \dots, x_n]$. For a variable $x_k$ with $1\le k \le n$, we write $f_1$ and $f_2$ as polynomials in $x_k$:
\begin{equation*}
    \begin{aligned}
        f_1 &~= a_s x_k^s + a_{s-1}x_k^{s-1} + \dots + a_0,\\
        f_2 &~= b_r x_k^r + b_{r-1}x_k^{r-1} + \dots + b_0.
    \end{aligned}
\end{equation*}
The resultant of $f_1$ and $f_2$ with respect to $x_k$, denoted $\textsf{Res}(f_1, f_2, x_k)$, is defined as
\begin{equation*}
\textsf{Res}(f_1, f_2, x_k) \defeq \begin{vmatrix}
a_s &~ a_{s-1} &~ \dots &~ a_0 &~ &~ &~ \\
 &~ a_s &~ a_{s-1} &~ \dots &~ a_0 &~ &~ \\
 &~ &~ \ddots &~ &~ &~ \ddots &~ \\
 &~ &~ &~ a_s &~ a_{s-1} &~ \dots &~ a_0 \\
b_r &~ b_{r-1} &~ \dots &~ b_0 &~ &~ &~ \\
 &~ b_r &~ b_{r-1} &~ \dots &~ b_0 &~ &~ \\
 &~ &~ \ddots &~ &~ &~ \ddots &~ \\
 &~ &~ &~ b_r &~ b_{r-1} &~ \dots &~ b_0
\end{vmatrix}.    
\end{equation*}
\end{definition}

\begin{definition}[Discriminant]
Let $f$ be a polynomial in $\mathbb{R}[x_1, \dots, x_n]$. For a variable $x_k$ with $1\le k \le n$, assume that
\begin{equation*}
    f_1 = a_s x_k^s + a_{s-1}x_k^{s-1} + \dots + a_0,
\end{equation*}
The discriminant of $f$ with respect to $x_k$, denoted  $\textsf{Dis}(f, x_k)$, is defined as
\begin{equation*}
    \textsf{Dis}(f, x_k) \defeq \tfrac{(-1)^{\tfrac{s(s-1)}{2}}}{a_s} \textsf{Res}\left(f, \tfrac{\partial f}{\partial x_k}, x_k\right).
\end{equation*}
\end{definition}

\begin{definition}[McCallum's Projection Operator~\cite{mccallum88jsc}]
Let $P\subset \Int[x_1,\dots,x_n]$ be a finite set of polynomials containing at least two variables.
Let $B$ be an irreducible basis of the primitive part of $P$.
The McCallum's projector operator, denoted $\textsf{Proj}(\cdot, \cdot)$, is defined as
\begin{equation*}
\begin{aligned}
    \textsf{Proj}(P,x_k) \defeq  &~~\set{\textsf{cont}(f,x_k)\given f\in P}\\
    \cup &~~\set{\textsf{coeff}(f,x_k)\given f\in B}\\
    \cup &~~\set{\textsf{dis}(f,x_k)\given f\in B}\\
    \cup &~~\set{\textsf{res}(f,g,x_k)\given f,g\in B, f\neq g}.
\end{aligned}
\end{equation*}
\end{definition}

\myparagraph{The projection phase.} 
Fix a variable elimination ordering from $x_n$ to $x_1$ if fixed (which is the reverse of the ordering used in our \cref{sec:pre}). The projection operator $\textsf{Proj}$ is applied recursively on $\mathcal{F}$ for $n-1$ times.
$$\begin{aligned}
\mathcal{F}_n &~\defeq \mathcal{F} \subset \Real[x_1,\dots, x_n]\\
\mathcal{F}_{n-1} &~\defeq \textsf{Proj}(\mathcal{F}_n, x_n) \subset \Real[x_1,\dots, x_{n-1}]\\
&~\vdots\\
\mathcal{F}_1 &~\defeq \textsf{Proj}(\mathcal{F}_{2}, x_{2}) \subset \Real[x_1]\\
\end{aligned}$$
In each step, the zero sets of the polynomials in $\mathcal{F}_{i-1}$ are the projection of ``significant points'' of the zero set of polynomial in $\mathcal{F}_i$, i.e. self crossings, isolated points, vertical tangent points etc.
 
\myparagraph{The lifting phase.} We first deal with the 1-dimensional case. Enumerate the real zeros of the monovariate polynomials in $\mathcal{F}_n$ as 
\begin{equation*}
    -\infty<\xi_1<\xi_2<\dots<\xi_s<+\infty,  
\end{equation*}
which gives an $\mathcal{F}_n$-invariant decomposition of $\Real$. The purpose of this phase is to isolate the above zeros and find sample points for each component in the decomposition. Note that for an open interval we may choose a rational sample point but for a zero we must store an exact representation of the algebraic number.

Then we lift a $\mathcal{F}_{i}$-invariant decomposition $\mathcal{D'}$ of $\Real^i$ to an $\mathcal{F}_{i+1}$-invariant decomposition $\mathcal{D}$ of $\Real^{i+1}$.
By construction, we know that $\mathcal{F}_{i+1}$ is delineable over each cell in $\mathcal{D'}$. For each cell $C\in D'\in \Real^i$, choose a sample point $(\hat{x}_1,\dots,\hat{x}_i)\in \Real^i$ and substitute the variables $(x_1,\dots,x_i)$ in $\mathcal{F}_{i+1}$ by these sampled values. In this way, the problem is again reduced to the 1-dimensional case and can be done similarly.

\subsection{The Formula in Running Example}\label{app:example}
\begin{align*}
\small
 \Phi &~\defeq \exists x_{1}, x_{2}, x_{3}, x_{4}, x_{5}, x_{6}, x_{7}, x_{8}.\\
 &~
 \begin{pmatrix}
 \small
 \begin{aligned}
    &~x_{1} + 2 x_{2} + 3 x_{3} \le 20  &~ \wedge &~ - x_{1} + x_{2} - 2 x_{3} \le 5 &~ \wedge &~ x_{1} - 4 x_{2} \le 0 \\
    \wedge &~3 x_{1} - 4 x_{2} - 3 x_{3} \le 3 &~ \wedge &~ 3 x_{1} + x_{3} \le -3 &~ \wedge &~ 3 x_{2} - 2 x_{3} + 4 x_{4} \le -5 \\
    \wedge &~5 x_{2} - x_{4} - 5 x_{5} \le 1 &~ \wedge &~ - 5 x_{2} - 3 x_{5} + 3 x_{8} \le 2 &~ \wedge &~ 4 x_{4} + x_{5} + x_{6} \le 1 \\
    \wedge &~3 x_{6} + 2 x_{7} \le -2 &~ \wedge &~ - 5 x_{3} \le -2 &~ \wedge &~ - x_{2} - x_{3} \le 1 \\
    \wedge &~x_{3} \le -1 &~ \wedge &~ x_{5} - 5 x_{6} \le -4 &~ \wedge &~ - 5 x_{4} + 5 x_{5} + 4 x_{6} \le -4 \\
    \wedge &~4 x_{2} - 4 x_{5} \le -5 &~ \wedge &~ 5 x_{2} + 3 x_{3} + 5 x_{4} \le 5 &~ \wedge &~ - 4 x_{2} + 3 x_{4} + 3 x_{5} \le -1 \\
    \wedge &~3 x_{2} \le -4 &~ \wedge &~- x_{6} - x_{7} \le 2 \\
 \end{aligned}
 \end{pmatrix}
\end{align*}

Eliminating the variable subset $\{x_1, x_2, x_3, x_4, x_5\}$ using the greedy algorithm follows the elimination order $x_1, x_4, x_3, x_2, x_5$ and yields $3,684$ inequalities. By contrast, our heuristic produces the elimination order $x_1, x_3, x_4, x_5, x_2$ and yields only $1,680$ inequalities, a reduction of 2,004 inequalities ($\approx 55.4\%$) compared with the greedy method.

\subsection{Extensions of the Dynamic Programming Framework}\label{app:algo}

\cref{algo:cad} presents the dynamic programming framework for the CAD's Projection Phase.
\cref{algo:qe} presents the algorithm for extracting a variable elimination order from a nice tree decomposition.

\begin{algorithm2e}[H]
\caption{A DP Framework for CAD's Projection Phase}
\label{algo:cad}
\SetKwInOut{Input}{Input}
\SetKwInOut{Output}{Output}
\Input{a NRA formula $\Phi$ of the form \cref{eq:input};}
\Output{a set of polynomials within $\Real[x_{m+1},\dots,x_n]$}
\BlankLine
\algocomment{Invoke external algorithms to compute nice tree-decompositions.}
Construct the primal graph~$G_\Phi=(V_\Phi,E_\Phi)$ of $\Phi$\;
Compute a nice tree decomposition~$T=(\mathcal{B},E_T)$ of~$G_\Phi$\;
\algocomment{Initialize $\mathcal{I}$ from top down.}
$P \subset \Real[\seq{x}] \gets$ the set of polynomials in $\Phi$\;
Initialize an empty queue of bags $q=\langle \rangle$\;
$q.\texttt{enqueue}(r_T)$ \;
\While{$q$ is not empty}{
    $b \gets q.\texttt{dequeue}()$\;
    $q.\texttt{enqueue}(b.\texttt{children})$\;
     \eIf{$b.\textnormal{\texttt{type}} = \textsf{leaf}$ }{
        $\mathcal{I}(b) \gets \emptyset$\;
    }{
        $\mathcal{I}(b) \gets \{p \in P \mid \texttt{var}(p)\subseteq b\}$\;
        $P \gets P\setminus \mathcal{I}(b)$\;
    }
}
\algocomment{Recursively compute $\mathcal{V}$ from bottom up.}
$q.\texttt{enqueue}(\mathcal{B}.\texttt{leaves})$ \;
\While{$q$ is not empty}{
    $b \gets q.\texttt{dequeue}()$\;
    $q.\texttt{enqueue}(b.\texttt{parent})$\;
    \Switch{$b.\textnormal{\texttt{type}}$}{
        \lCase{\textsf{join}}{ $\mathcal{V}(b)\gets \mathcal{I}(b)\cup \mathcal{V}(b.\texttt{leftchild}) \cup \mathcal{V}(b.\texttt{rightchild})$}
        \lCase{\textsf{introduce}}{$\mathcal{V}(b) \gets \mathcal{I}(b)\cup \mathcal{V}(b.\texttt{children})$}
        \lCase{\textsf{forget}}{$\mathcal{V}(b) \gets \mathcal{I}(b)\cup \textsf{Proj} (\bigwedge\mathcal{V}(b.\texttt{children}), b.\texttt{forget})$}
        \lCase{\textsf{leaf}}{$\mathcal{V}(b)\gets \mathcal{I}(b)$}
    }
}
\Return{$\bigcup \mathcal{V}(r_T)$}\;\algocomment{After this algorithm finishes, the CAD algorithm continuous the projection phase for free variables, until only one variable is left. Then the lifting phase begins.}
\end{algorithm2e}

\begin{algorithm2e}[t]
\SetKwInOut{Input}{Input}
\SetKwInOut{Output}{Output}
\Input{a $\mathcal{T}$-formula $\Phi$ of the form \cref{eq:input} and a QE procedure $\textsf{QE}$;}
\Output{a quantifier-free $\mathcal{T}$-formula equivalent to $\Phi$}
\BlankLine
Construct the primal graph~$G_\Phi=(V_\Phi,E_\Phi)$ of $\Phi$\;
Compute a (nice) tree decomposition~$T=(\mathcal{B},E_T)$ of~$G_\Phi$\;
\algocomment{Extract a variable elimination order from $T$.}
Initialize an empty queue of variables $order = \langle \rangle$\; 
Initialize an empty queue of bags $q = \langle \rangle$\;
$q.\texttt{enqueue}(r_T)$\;
\While{$q$ is not empty}{
    $b \gets \texttt{dequeue}(q)$\;
    $\texttt{enqueue}(q, b.\texttt{children})$\;
    \eIf{$b$ is root}{
        $order.\texttt{enqueue}(b)$\;
    }{
        $order.\texttt{enqueue}(b\setminus b.\textsf{parent})$
    }
}
$order\gets order.\texttt{reverse}$ \tcp*{reverse the ordering}
\Return{$\textnormal{\textsf{QE}}(\Phi, \text{order})$}
\caption{A General Framework for QE Procedures}\label{algo:qe}
\end{algorithm2e}

\subsection{Sparsity in SMT-LIB Benchmarks}\label{app:smtlib}

We examined the \texttt{QF\_LRA} and \texttt{QF\_NRA} categories from the SMT-LIB benchmark library~\cite{barrett16smtlib}. 
Using \texttt{pysmt}, we successfully parsed 1753 LRA instances and 9249 NRA instances. 
Among the parsed instances, only 337 LRA instances and 807 NRA instances have non-complete primal graphs.
We first analyzed the treewidth of these instances, the results are presented in \cref{tab:tw-lra} and \cref{tab:tw-nra}.
However, since many instances contain hundreds or more variables, it is impractical to perform a full tree decomposition for every instance (e.g., within a 10-minute timeout). 
Therefore, we also examined the edge density sparsity, defined by $|E|/\binom{|V|}{2}$, where $E$ and $V$ denote the number of edges and vertices of a primal graph, respectively.
The results are presented in \cref{tab:edge-lra} and \cref{tab:edge-nra}.
From these results, we found that SMT-LIB benchmarks are generally too large and too dense for demonstrating our algorithms.
\vspace{-0.5cm}

\begin{table}[H]
    \centering
    \caption{Treewidth distribution among 337 sparse instances in the \texttt{QF\_LRA} dataset}
    \label{tab:tw-lra}
    \resizebox{\linewidth}{!}{
    \begin{tabular}{c|c|c|c|c|c}
        \hline
        Treewidth range &~ Number &~ \% in sparse instances &~ Min \#var &~ \ Max \#var &~ Median \#var \\
        \hline
            (0, 10] & 1 & 0.27 & 70 & 70 & 70.0 \\
        \hline
            (10, 20] & 47 & 12.47 & 56 & 8954 & 144.0 \\
        \hline
            (20, 40] & 98 & 25.99 & 102 & 9937 & 2019.5 \\
        \hline
            (40, 60] & 31 & 8.22 & 259 & 7242 & 1307.0 \\
        \hline
            (60, 80] & 21 & 5.57 & 223 & 4467 & 1173.0 \\
        \hline
            (80, 100] & 21 & 5.57 & 315 & 7793 & 1433.0 \\
        \hline
            (100, 200] & 58 & 15.38 & 425 & 4855 & 2065.0 \\
        \hline
            (200, 250] & 5 & 1.33 & 1291 & 2452 & 1881.0 \\
        \hline
            Timeout & 95 & 0.25 & NA & NA & NA \\
        \hline
    \end{tabular}
    }
\vspace{-0.5cm}
\end{table}

\begin{table}[!t]
    \centering
    \caption{Treewidth distribution among 807 sparse instances in the \texttt{QF\_NRA} dataset}
    \label{tab:tw-nra}
    \resizebox{\linewidth}{!}{
    \begin{tabular}{c|c|c|c|c|c}
        \hline
        Treewidth range &~ Number &~ \% in sparse instances &~ Min \#var &~ \ Max \#var &~ Median \#var \\
        \hline
            (0, 5] & 11 & 1.36 & 6 & 248 & 6.0 \\
        \hline
            (5, 10] & 7 & 0.87 & 76 & 424 & 330.0 \\
        \hline
            (10, 20] & 20 & 2.48 & 33 & 8220 & 255.5 \\
        \hline
            (20, 40] & 294 & 36.43 & 37 & 9189 & 467.0 \\
        \hline
            (40, 60] & 153 & 18.96 & 68 & 9352 & 680.0 \\
        \hline
            (60, 80] & 109 & 13.51 & 102 & 6068 & 887.0 \\
        \hline
            (80, 100] & 53 & 6.57 & 126 & 3441 & 1167.0 \\
        \hline
            (100, 200] & 89 & 11.03 & 124 & 5149 & 1183.0 \\
        \hline
            (200, 300] & 7 & 0.87 & 279 & 1541 & 380.0 \\
        \hline
            (300, 550] & 2 & 0.25 & 605 & 756 & 680.5 \\
        \hline
            Timeout & 62 & 0.08 & NA & NA & NA \\
        \hline
    \end{tabular}
    }
\vspace{-0.3cm}
\end{table}

\begin{table}[!t]
    \centering
    \caption{Edge density among 337 sparse instances in the \texttt{QF\_LRA} dataset}
    \label{tab:edge-lra}
    \resizebox{\linewidth}{!}{
    \begin{tabular}{c|c|c|c|c|c}
        \hline
        Density range &~ Number &~ \% in sparse instances &~ Min \#var &~ \ Max \#var &~ Median \#var \\
        \hline
            (0.0, 0.05] & 225 & 59.68 & 307 & 9937 & 2864.0 \\
        \hline
            (0.05, 0.1] & 50 & 13.26 & 147 & 9385 & 365.5 \\
        \hline
            (0.1, 0.15] & 37 & 9.81 & 83 & 8622 & 204.0 \\
        \hline
            (0.15, 0.2] & 35 & 9.28 & 61 & 6416 & 132.0 \\
        \hline
            (0.2, 0.25] & 6 & 1.59 & 56 & 444 & 174.0 \\
        \hline
            (0.3, 0.35] & 8 & 2.12 & 1617 & 1617 & 1617.0 \\
        \hline
            (0.85, 0.9] & 1 & 0.27 & 1785 & 1785 & 1785.0 \\
        \hline
            (0.9, 0.95] & 4 & 1.06 & 1274 & 4985 & 2435.0 \\
        \hline
            (0.95, 1.0] & 11 & 2.92 & 1798 & 4285 & 1798.0 \\
        \hline
    \end{tabular}
    }
\vspace{-0.3cm}
\end{table}

\begin{table}[!t]
    \centering
    \caption{Edge density among 807 sparse instances in the \texttt{QF\_NRA} dataset}
    \label{tab:edge-nra}
    \resizebox{\linewidth}{!}{
    \begin{tabular}{c|c|c|c|c|c}
        \hline
        Density range &~ Number &~ \% in sparse instances &~ Min \#var &~ \ Max \#var &~ Median \#var \\
        \hline
            (0.0, 0.05] &~ 418 &~ 51.80 &~ 122 &~ 9929 &~ 2025.0 \\
            \hline
            (0.05, 0.1] &~ 189 &~ 23.42 &~ 123 &~ 9525 &~ 441.0 \\
            \hline
            (0.1, 0.15] &~ 50 &~ 6.20 &~ 57 &~ 3309 &~ 215.0 \\
            \hline
            (0.15, 0.2] &~ 23 &~ 2.85 &~ 28 &~ 4052 &~ 249.0 \\
            \hline
            (0.2, 0.25] &~ 2 &~ 0.25 &~ 154 &~ 154 &~ 154.0 \\
            \hline
            (0.25, 0.3] &~ 4 &~ 0.50 &~ 110 &~ 756 &~ 190.0 \\
            \hline
            (0.3, 0.35] &~ 15 &~ 1.86 &~ 68 &~ 282 &~ 168.0 \\
            \hline
            (0.35, 0.4] &~ 9 &~ 1.12 &~ 70 &~ 282 &~ 130.0 \\
            \hline
            (0.4, 0.45] &~ 16 &~ 1.98 &~ 59 &~ 1678 &~ 175.0 \\
            \hline
            (0.45, 0.5] &~ 14 &~ 1.73 &~ 70 &~ 216 &~ 95.5 \\
            \hline
            (0.5, 0.55] &~ 17 &~ 2.11 &~ 51 &~ 303 &~ 86.0 \\
            \hline
            (0.55, 0.6] &~ 16 &~ 1.98 &~ 33 &~ 876 &~ 115.5 \\
            \hline
            (0.6, 0.65] &~ 9 &~ 1.12 &~ 68 &~ 1136 &~ 208.0 \\
            \hline
            (0.65, 0.7] &~ 4 &~ 0.50 &~ 6 &~ 178 &~ 104.5 \\
            \hline
            (0.7, 0.75] &~ 5 &~ 0.62 &~ 172 &~ 702 &~ 605.0 \\
            \hline
            (0.75, 0.8] &~ 9 &~ 1.12 &~ 6 &~ 1925 &~ 6.0 \\
            \hline
            (0.8, 0.85] &~ 2 &~ 0.25 &~ 124 &~ 198 &~ 161.0 \\
            \hline
            (0.85, 0.9] &~ 3 &~ 0.37 &~ 132 &~ 3525 &~ 2602.0 \\
            \hline
            (0.9, 0.95] &~ 2 &~ 0.25 &~ 3802 &~ 5125 &~ 4463.5 \\
            \hline
    \end{tabular}
    }
\end{table}

\subsection{Additional Experiments} \label{app:experiments}

We conducted experiments on benchmarks of the lattice reachability problems from~\cite[Example.~9]{li23jsc}, and the results are presented in \cref{tab:lat}. 
Note that, since the primal graphs exhibit a special "straight-line" structure, the variable elimination orderings generated by our methods coincide with those of [51], and hence the experimental results also coincide. 

\begin{table}[!h]
    \centering
    \caption{CAD performance on lattice reachability problems over \texttt{Mathematica}}
    \label{tab:lat}
    \begin{tabular}{c|c r|c|c|c}
        \hline
        ID & \multicolumn{2}{c|}{Instance} & SVO & Brown & PEO/Ours \\
        \hline

        \multirow{2}{*}{13} & \multirow{2}{*}{\makecell{\#var=11\  \#ineq=8\\ \#max\_deg=2\ \#tw=3}} & Cells & \textbf{5,299} & 6,067 & \textbf{5,299} \\
        \cline{4-6}
        & & Time & 2.5 & \textbf{1.1} & 2.5 \\
        \hline
        
        \multirow{2}{*}{14} & \multirow{2}{*}{\makecell{\#var=12\  \#ineq=9\\ \#max\_deg=2\ \#tw=3}} & Cells & \textbf{14,509} & 24,141 & \textbf{14,509} \\
        \cline{4-6}
        & & Time & 6.8 & \textbf{5.4} & 7.4 \\
        \hline

        \multirow{2}{*}{15} & \multirow{2}{*}{\makecell{\#var=13\  \#ineq=10\\ \#max\_deg=2\ \#tw=3}} & Cells & \textbf{39,625} & 65,961 & \textbf{39,625} \\
        \cline{4-6}
        & & Time & 20.8 & \textbf{15.8} & 23.4 \\
        \hline

        \multirow{2}{*}{16} & \multirow{2}{*}{\makecell{\#var=14\  \#ineq=11\\ \#max\_deg=2\ \#tw=3}} & Cells & 127,195 & 185,515 & \textbf{109,787} \\
        \cline{4-6}
        & & Time & \textbf{25.6} & 43.5 & 53.6 \\
        \hline
        
        \multirow{2}{*}{17} & \multirow{2}{*}{\makecell{\#var=15\  \#ineq=12\\ \#max\_deg=2\ \#tw=3}} & Cells & 511,703 & \textbf{303,967} & \textbf{303,967} \\
        \cline{4-6}
        & & Time & 175.1 & 194.3 & \textbf{142.0} \\
        \hline
        
        \multirow{2}{*}{18} & \multirow{2}{*}{\makecell{\#var=16\  \#ineq=13\\ \#max\_deg=2\ \#tw=3}} & Cells & 3,113,905 & \textbf{847,351} & \textbf{847,351}  \\
        \cline{4-6}
        & & Time & 1565.2 & \textbf{430.7} & 432.9 \\
        \hline
        
        \multirow{2}{*}{19} & \multirow{2}{*}{\makecell{\#var=17\  \#ineq=14\\ \#max\_deg=2\ \#tw=3}} & Cells & NA & \textbf{2,362,235} & \textbf{2,362,235}  \\
        \cline{4-6}
        & & Time & $>3\text{h}$ & 1436.6 & \textbf{1399.5} \\
        \hline
    \end{tabular}
\end{table}

\subsection{The NRA Formulas in \cref{tab:cad}}

The randomly generated LRA formulas can be found via the link to our code.
The randomly generated sparse NRA formulas in \cref{tab:cad} are shown here.

\begin{align*}
 \Phi_7 &~\defeq \exists x_{1}, x_{2}, x_{3}, x_{4}, x_{5}, x_{6}.\\
 &~
 \begin{pmatrix}
 \begin{aligned}
    &~ 8 x_1x_2 + 6 x_1 + 5 x_3 + 5 \ge 0 \\
    \wedge&~ -4 x_1x_3 -10 x_4^{2} + 4 x_4 + 9 \ge 0 \\
    \wedge&~ 6 x_3x_5 -3 x_4x_5 + 4 x_4^{2} + 8 x_5^{2} -6 \ge 0 \\
    \wedge&~ -5 x_4x_5 + 7 x_4x_6 + x_5^{2} -5 \ge 0 \\
    \wedge&~ -2 x_2^{2} + 6 \ge 0 \\
    \wedge&~ -10 x_1x_3 -4 x_3^{2} + 5 x_2 -7 x_3 + 2 \ge 0 \\
 \end{aligned}
 \end{pmatrix}
\end{align*}

\begin{align*}
 \Phi_8 &~\defeq \exists x_{1}, x_{2}, x_{3}, x_{4}, x_{5}, x_{6}, x_{7}.\\
 &~
 \begin{pmatrix}
 \begin{aligned}
    &~ -8 x_1x_2^{2} + 5 x_1x_3^{2} + x_4^{2} + 7 \ge 0 \\
    \wedge&~ 7 x_1x_2x_4 -9 x_2^{2} -2 x_5^{2} + 5 \ge 0 \\
    \wedge&~ -2 x_1^{2}x_6 -6 x_4x_5 -8 \ge 0 \\
    \wedge&~ -9 x_2x_7^{2} + 6 x_1x_4 -9 \ge 0 \\
 \end{aligned}
 \end{pmatrix}
\end{align*}

\begin{align*}
 \Phi_9 &~\defeq \exists x_{1}, x_{2}, x_{3}, x_{4}, x_{5}, x_{6}, x_{7}.\\
 &~
 \begin{pmatrix}
 \begin{aligned}
    &~ 7 x_{1}x_{3} + 2 x_{2}x_{3} + 7 x_{2}^{2} -8 x_{3}^{2} + 4 x_{2} + 10 x_{3} \ge 0 \\
    \wedge&~ -5 x_{2}x_{3} -x_{3}x_{5} + 7 x_{5}^{2} -3 x_{2} \ge 0 \\
    \wedge&~ 7 x_{2}^{2} + 3 x_{4}^{2} + 4 x_{5}^{2} -10 x_{4} + 6 x_{5} \ge 0 \\
    \wedge&~ -x_{1}x_{6} + 2 x_{1}^{2} + 4 x_{6} \ge 0 \\
    \wedge&~ 9 x_{4}x_{7} -10 x_{4} \ge 0 \\
 \end{aligned}
 \end{pmatrix}
\end{align*}

\begin{align*}
 \Phi_{10} &~\defeq \exists x_{1}, x_{2}, x_{3}, x_{4}, x_{5}, x_{6}, x_{7}, x_{8}.\\
 &~
 \begin{pmatrix}
 \begin{aligned}
    &~ 10 x_{1}^{2}x_{2}x_{3} + 8 x_{1}^{4} -7 x_{1}x_{2}^{2} + 8 x_{1}^{2}x_{2} + 10 x_{1}^{2} + 10 x_{2}^{2} \ge 0 \\
    \wedge&~ 5 x_{3}x_{4}^{3} + 10 x_{4}^{4} + 4 x_{3}x_{5}^{2} + 2 x_{5}^{3} \ge 0 \\
    \wedge&~ 5 x_{3}x_{4}^{2}x_{7} -4 x_{3}^{2}x_{4} -9 x_{4}^{3} + 4 x_{3}x_{7} + 4 x_{3} + 8 x_{7} \ge 0 \\
    \wedge&~ 2 x_{4}x_{5}^{2}x_{6} + 2 x_{4}^{3}x_{6} + 3 x_{4}^{2}x_{5} + 4 x_{4}^{3} + 7 x_{4}x_{6} -10 x_{5} \ge 0 \\
    \wedge&~ 10 x_{4}x_{8}^{3} -3 x_{4}x_{8}^{2} -9 x_{4}^{3} + 3 x_{8}^{2} \ge 0 \\
 \end{aligned}
 \end{pmatrix}
\end{align*}

\begin{align*}
 \Phi_{11} &~\defeq \exists x_{1}, x_{2}, x_{3}, x_{4}, x_{5}, x_{6}, x_{7}, x_{8}.\\
 &~
 \begin{pmatrix}
 \begin{aligned}
    &~ 10 x_1x_2 + 2 x_2x_3 + 4 x_3x_4 + 6 \ge 0 \\
    \wedge&~ -x_1x_3 -9 x_2x_3 + 9 x_2x_5 + 9 \ge 0 \\
    \wedge&~ 3 x_1x_6 -7 x_2x_4 -2 x-6^{2} -3 \ge 0 \\
    \wedge&~ -5 x_2x_6 + 7 x_2x_7 + 8 x-4x6 + 4 x_4 -10 x_7 -9 \ge 0 \\
    \wedge&~ -6 x_1x_4 -3 x_3x_4 + 3 x_4x_8 + 2 x_1 -10 \ge 0 \\
    \wedge&~ 7 x_2^{2} -2 \ge 0 \\
 \end{aligned}
 \end{pmatrix}
\end{align*}

\begin{align*}
 \Phi_{12} &~\defeq \exists x_{1}, x_{2}, x_{3}, x_{4}, x_{5}, x_{6}, x_{7}, x_{8}, x_{9}.\\
 &~
 \begin{pmatrix}
 \begin{aligned}
    &~ -x_1x_4 + 5 x_1^{2} + x_2x_3 + 9 \ge 0 \\
    \wedge&~ -8 x_1^{2} -5 x_2^{2} -8 x_3x_5 -10 x_3 + 7 \ge 0 \\
    \wedge&~ -9 x_1x-2 -6 x_1x_5 -3 x_6^{2} -9 x_1 -9 \ge 0 \\
    \wedge&~ 7 x_1x_7 + 10 x_2x_5 + 4 x_2^{2} -3 \ge 0 \\
    \wedge&~ -8 x_1x_2 + 7 x_5x_8 + 8 \ge 0 \\
    \wedge&~ 9 x_3x_5 + 8 x_3x_9 + 2 x_1 -5 \ge 0 \\
    \wedge&~ 2 x_1^{2} -2 x_2^{2} + 5 x_5x_7 + 1 \ge 0 \\
    \wedge&~ x_2x_4 + 1 \ge 0 \\
    \wedge&~ -3 x_3^{2} + 5 \ge 0 \\
    \wedge&~ -3 x_1x_2 -6 x_1^{2} -4 \ge 0 \\
 \end{aligned}
 \end{pmatrix}
\end{align*}